\documentclass[onecolumn,draftcls,11pt]{IEEEtran}
\usepackage{cite}

\usepackage[dvips]{graphicx}

\usepackage[cmex10]{amsmath}
\usepackage{amssymb}
\usepackage{bm}
\usepackage{bbm}

\usepackage{array}
\newcommand{\vva}{``}
\newcommand{\mc}{\mathcal}
\newcommand{\titolo}{Theoretical Performance Analysis of Eigenvalue-based Detection}
\newcommand{\beq}{\begin{equation}}
\newcommand{\eeq}{\end{equation}}
\newcommand{\convas}{\stackrel{\text{\scriptsize a.s.}}{\longrightarrow}}
\newcommand{\convd}{\stackrel{\text{\scriptsize{\textit{D}}}}{\longrightarrow}}
\newcommand{\ra}{\rightarrow}

\newtheorem{theorem}{Theorem}[section]

\newenvironment{proof}[1][Proof]{\begin{trivlist}
\item[\hskip \labelsep {\bfseries #1}]}{\end{trivlist}}

\newcommand{\qed}{\nobreak \ifvmode \relax \else
      \ifdim\lastskip<1.5em \hskip-\lastskip
      \hskip1.5em plus0em minus0.5em \fi \nobreak
      \vrule height0.75em width0.5em depth0.25em\fi}

\begin{document}
%
\title{\titolo}
%
%
%

\author{Federico Penna,~\IEEEmembership{Student Member,~IEEE,}
        and~Roberto Garello,~\IEEEmembership{Senior~Member,~IEEE}

\thanks{F. Penna is with TRM Lab - Istituto Superiore Mario Boella (ISMB) and with the Department
of Electrical Engineering (DELEN), Politecnico di Torino, Italy. R. Garello is with the Department
of Electrical Engineering (DELEN), Politecnico di Torino, Italy.
e-mail: \{federico.penna, roberto.garello\}@polito.it.}
}
\maketitle

\begin{abstract}
In this paper we develop a complete analytical framework based on Random Matrix Theory for the performance evaluation of Eigenvalue-based Detection.
While, up to now, analysis was limited to false-alarm probability, we have obtained an analytical expression also for the probability of missed detection, by using the theory of spiked population models. A general scenario with multiple signals present at the same time is considered.
The theoretical results of this paper allow to predict the error probabilities, and to set the decision threshold accordingly, by means of a few mathematical formulae.
In this way the design of an eigenvalue-based detector is made conceptually identical to that of a traditional energy detector. 
As additional results, the paper discusses the conditions of signal identifiability
for single and multiple sources. 
All the analytical results are validated through numerical simulations, covering also convergence, identifiabilty and non-Gaussian practical modulations.  
\end{abstract}


%
\IEEEpeerreviewmaketitle

\section{Introduction}
%
%
%
%
Eigenvalue-based Detection (EBD) has been introduced \cite{debbah,liang} as an efficient technique to perform spectrum sensing in Cognitive Radio (CR).
Using the EDB approach, the secondary receiver is able to infer the presence or the absence of a primary signal based on the largest and the smallest eigenvalue of the received signal's covariance matrix. This technique requires a cooperative detection setting, which may be accomplished by multiple antennas or cooperation among different users.  In addition to the CR context, the detection of signal components in noisy covariance matrices is a very general problem, with a wide variety of applications in communications, statistics, genetics, mathematical finance, artificial learning.

The main advantage offered by EDB is its robustness to the problem of noise uncertainty, which affects all the previously proposed detection schemes including the widely adopted Energy Detection (ED).
However, while for ED there exist comprehensive theoretical results that allow to express the error probabilities through analytical formulae, a corresponding theory for EBD has not been fully developed yet.

In general, a signal detection scheme can be characterized by defining two types of error probabilities: the probability of \textit{false alarm} and the probability of\textit{ missed detection} (see Sec. \ref{model} for a formal definition). These probabilities depend on the \textit{decision threshold} (the value used by the algorithm to decide whether a signal is present or absent). If analytical formulae are available, it is possible to:
\begin{enumerate}
	\item[\textit{a)}] \textit{predict} the error probabilities of the system as a function of the decision threshold;
	\item[\textit{b)}]	\textit{set} the decision threshold according to the required error constraints.
\end{enumerate}
Such formulae are well-known in case of ED. For EBD, up to now, only approximated criteria were proposed for the estimation of the false-alarm probability \cite{debbah, liang} and, to the best of our knowledge, no exact analytical results have been found for the missed-detection probability yet.

In this paper, by exploiting the spectral properties of the sample covariance matrix under the two complementary conditions of signal present/absent, we derive analytical expressions both for the false-alarm and the missed-detection probability.
The result is a complete probabilistic framework that allows to evaluate the performance of EBD and to determine the proper decision threshold through analytical formulae.

Whereas most of the works on detection consider only the case of a \textit{single} signal to be detected, our results also apply to the case of \textit{multiple} primary signals. This generalization is of interest for the applications in CR, since a secondary user might be located in such a way as to hear different primary signals (each with a different channel). The analytical results derived in this paper show that the number of signals simultaneously present, as well as their powers and their channels, have an impact on the detection performance.

The paper is organized as follows: Sec \ref{maxmin} introduces the signal model and the theoretical foundations of eigenvalue-based detection; Sec. \ref{fa} and \ref{md} derive analytical results for the probabilities of false alarm and missed detection, and for the signal identifiability condition; Sec. \ref{discussion} discusses the problem of setting a proper decision threshold; Sec. \ref{results} validates the analysis through  numerical results; Sec. \ref{concl} concludes.

%

\section{Eigenvalue-based detection}
\label{maxmin}

\paragraph*{Notational remark}
In the following, upper-case boldface letters indicate matrices, lower-case boldface letters indicate vectors, the symbols $^T$ and $^H$ indicate respectively the transpose and conjugate transpose (Hermitian) operators, $\mathrm{tr}(\cdot)$ is the trace of a matrix, $\| \cdot \|$ is the Euclidean norm of a vector, $\mathrm{diag}(\bm{x})$ indicates a square diagonal matrix whose main diagonal entries are taken from the vector $\bm{x}$, $\bm{I}_{N}$ is the identity matrix (of size $N$ if specified), $\bm{0}_{M,N}$ is a $M \times N$ matrix of zeros; the symbol $\triangleq$ stands for \vva defined as'', the symbol $\sim$ for \vva distributed with law'', $\convas$ indicates the almost sure convergence, and $\convd$ the convergence in distribution; $I_{\{\alpha\}}$ is the indicator function which takes value $1$ where the condition $\alpha$ is true and $0$ elsewhere.
\subsection{Signal model}
\label{model}

We consider a \textit{cooperative detection} framework in which $K$ receivers (or antennas) collaborate to sense the spectrum.
Denote with $y_k$ be the discrete baseband complex sample at receiver $k$, and define the $K \times 1$ vector $\bm{y}=\left[y_1 \hdots y_K\right]^T$ containing the $K$ received signal samples.

The goal of the detector is to discriminate between two hypotheses:
\begin{itemize}
\item $\mathcal H_0$ (absence of primary signal). The samples contain only noise:
\begin{equation}
\bm{y}|_{\mathcal{H}_0} = \bm{v}
\end{equation}
where $\bm{v} \sim \mc{N}_{\mathbb{C}} (\bm{0}_{K,1}, \sigma^2_v \bm{I}_K)$ is a vector of circularly symmetric complex Gaussian (CSCG) noise samples;
\item $\mathcal H_1$ (presence of primary signal). For sake of generality, we consider a model where $P$ primary signals may be simultaneously present:
\begin{equation}
\bm{y}|_{\mathcal{H}_1} = \bm{H s} + \bm{v}
\end{equation}
where:  $\bm{H}$ is a $K \times P$ complex matrix, where each element $h_{kp}$ represents the channel between primary user $p$ and receiver $k$ (for simplicity, channels are assumed to be memoryless and constant for the sensing duration); $\bm{s}$ is a $P \times 1$ vector containing the primary signal samples, each coming from one of the $P$ sources. The primary signals are assumed to be complex, zero-mean and mutually independent with covariance matrix
\begin{equation}
\mathrm{E} \; \bm{s} \bm{s}^H  \triangleq \bm{\Sigma} = \mathrm{diag}(\sigma_1^2, \hdots, \sigma_P^2)
\end{equation}
where $\sigma^2_p$ is the variance of the $p$-th primary signal.
\end{itemize}
Under $\mc H_1$, we define the signal-to-noise ratio (SNR) as
\begin{equation}
\label{snr}
\rho \triangleq \frac{\mathrm{E} \; \| \bm{H s} \| ^2 } {\mathrm{E} \; \| \bm{v} \| ^2}
\end{equation}
This amounts to
	\begin{equation}
	\label{snrP}
	\rho =\frac{ \mathrm{tr} \; \bm{H} \bm{\Sigma} \bm{H}^H }{K \sigma^2_v} = \frac{\sum_{p=1}^P \sigma^2_p \| \bm{h}_p \|^2}{K \sigma^2_v}
	\end{equation}
where $\bm{h}_p$ is the $p$-th column of the matrix $\bm{H}$, i.e., the channel vector referred to primary source $p$.

In the \textit{single-user case} ($P=1$), we can drop the index $p$ and the expression of the SNR simplifies to
	\begin{equation}
	\label{snr1}
	\rho|_{P=1} = \frac{\sigma^2 \| \bm{h} \|^2}{K \sigma^2_v}
	\end{equation}

\subsubsection*{Remark}
All throughout this paper it is assumed that $P<K$. When this assumption is not verified, the covariance matrix lacks the necessary degrees of freedom to be able to distinguish the signal components from the noise. Notice that $P$ might be unknown, but to ensure a reliable detection $K$ (which is a receiver parameter) has to be chosen greater than the maximum possible number of primary signals.

\subsection{Spectral properties of the statistical covariance matrix}

Define the statistical covariance matrix of the received signal
\begin{equation}
 \bm{R} \triangleq \mathrm{E} \; \bm{y} \bm{y}^H
\end{equation}
Under $\mc H_0$ and $\mc H_1$ it is equal to, respectively
\begin{equation}
\label{R01}
\bm{R} = \left\{\begin{array}{ll}
 \sigma^2_v \bm{I}_K & (\mc H_0) \\
\bm{H} \bm{\Sigma} \bm{H}^H + \sigma^2_v \bm{I}_K   & (\mc H_1)
\end{array}\right.
\end{equation}
Let $\lambda_1 \geq \hdots \geq \lambda_K$ be the eigenvalues of $\bm{R}$ (without loss of generality, sorted in decreasing order).

Under $\mc H_0$, it is immediate to verify that
\beq
\label{l0}
\lambda_i|_{\mc H_0} = \sigma^2_v \; \; \; \forall i = 1, \hdots, K
\eeq

Under $\mathcal{H}_1$, there are $(K-P)$ eigenvalues equal to $\sigma^2_v$ and $P$ greater, since $\bm{H} \bm{\Sigma} \bm{H}^H$ is positive-semidefinite with rank $P$.
The eigenvalues in this case can be written as
\beq
\label{lambdai}
\lambda_i|_{\mc H_1} = \left\{\begin{array}{ll}
s_i + \sigma^2_v & (1 \leq i \leq P) \\
\sigma^2_v       & (P < i \leq K)
\end{array}\right.
\end{equation}
where $s_1 \geq \hdots \geq s_P > 0$ denote the $P$ non-zero eigenvalues of the \vva signal covariance matrix'' $\bm{H \Sigma H}^H$, and are found by solving the characteristic equation
\beq
\label{poly1}
\begin{array}{l}
\mathrm{det}\left(\bm{H \Sigma H}^H - s \bm{I}_K \right)=0 \\
 \mbox{s.t. } s \neq 0
\end{array}
\eeq
Because of the assumption $P<K$, the rank of the signal covariance matrix is $P$. It is possible to reduce the degree of the characteristic polynomial down to $P$ by applying the generalized Matrix Determinant Lemma (MDL) \cite{matref}
\begin{eqnarray}
&& \mathrm{det}\left(\bm{H \Sigma H}^H - s \bm{I}_K \right) = \nonumber \\
&& = \mathrm{det}(\bm{\Sigma}) \; \mathrm{det}(-s\bm{I}_K ) \; \mathrm{det}\left(\bm{\Sigma}^{-1} - \frac{1}{s} \bm{H}^H \bm{H} \right) = \nonumber \\
&& = \left( \prod_{p=1}^P \sigma_p^2 \right)  (-s)^{K-P} \; \mathrm{det}\left(\bm{H}^H \bm{H} - s \ \bm{\Sigma}^{-1} \right)
\label{mdl}
\end{eqnarray}
We note that the left-hand factor in (\ref{mdl}) is a constant with respect to $s$, the middle term gives rise to the $(K-P)$ trivial solutions $s=0$, while the right-hand term determines the non-zero roots. The signal eigenvalues $s_1 , \hdots , s_P$ may therefore be calculated from the simplified characteristic equation
\beq
\label{sP}
\mathrm{det}\left(\bm{H}^H \bm{H} - s \ \bm{\Sigma}^{-1} \right) = 0
\eeq
which has degree $P$ instead of $K$. Since $\bm{\Sigma}$ is diagonal, $\bm{\Sigma}^{-1} = \mathrm{diag}(\sigma_1^{-2}, \hdots, \sigma_P^{-2})$.

In the case of \textit{single primary user} ($P = 1)$,
there is one single signal eigenvalue
and, from (\ref{sP}), it has a very simple expression:
\beq
\label{s1}
s_1|_{P=1} = \| \bm{h} \|^2 \sigma^2
\eeq
where the index has been dropped like in (\ref{snr1}).
%

The spectral properties of $\bm R$, summarized by (\ref{l0}) and (\ref{lambdai}), motivate the adoption of the \textit{ratio between the largest and the smallest eigenvalue} of the covariance matrix as a test statistic to discriminate between the two hypotheses:
under $\mc H_0$ the ratio is equal to $1$, under $\mc H_1$ it is greater. This detection scheme was first proposed in \cite{debbah, liang}.

\subsection{Sample covariance matrix}

In practice, the statistical correlation matrix $\bm{R}$ is estimated through a \textit{sample covariance matrix}. Introduce $N$ as the number of samples collected by each receiver during the sensing period. It is assumed that consecutive samples are uncorrelated and that all the random processes involved (signals and noise) remain stationary for the sensing duration.
Then, let $\bm{s}(n)$, $\bm{v}(n)$ and $\bm{y}(n)$ be, respectively, the transmitted signal vector, the noise vector and the received signal vector at time $n$; define the $P \times N$ matrix
\beq
\bm{S} \triangleq \left[\bm{s}(1) \hdots \bm{s}(N)\right]
\eeq
and the $K \times N$ matrices
\begin{eqnarray}
\label{YN}
&\bm{V} \triangleq \left[\bm{v}(1) \hdots \bm{v}(N)\right]  \\
&\bm{Y} \triangleq \left[\bm{y}(1) \hdots \bm{y}(N)\right] = \bm{H S} + \bm{V}
\end{eqnarray}
The $K \times K$ sample covariance matrix $\bm{R}(N)$ is then defined as
\begin{equation}
\bm{R}(N) \triangleq \frac{1}{N} \bm{Y}\bm{Y}^H
\end{equation}
Denoting with $\hat{\lambda}_1 \geq \hdots \geq \hat{\lambda}_K$ its eigenvalues, the test statistic used for detection is
\begin{equation}
\label{T}
T \triangleq \frac{\hat{\lambda}_1}{\hat{\lambda}_K}
\end{equation}
Although $\bm{R}(N)$ converges to $\bm{R}$ as $N$ tends to infinity, for finite $N$ its properties depart from those of the statistical covariance matrix. In typical sensing applications $N$ is expected to be quite large (to increase the detection reliability) but still not enormous (to reduce the sensing time).
With such realistic values of $N$, the eigenvalues have no longer a deterministic behavior as in (\ref{R01}), but are characterized by a \textit{probability distribution}. Therefore the discrimination criterion based on the eigenvalues is not as sharp-cutting as in the ideal case and may be affected by two possible error events: \textit{false alarms} and \textit{missed detections}. Denoting with $\gamma$ the \textit{decision threshold} employed by the detector, such that
\beq
\mbox{decision} = \left\{ \begin{array}{ll} \mc H_0 & \mbox{if  } \, T<\gamma \\ \mc H_1 & \mbox{if   } \, T \geq \gamma \end{array} \right. \mbox{,} \nonumber
\eeq
the probability of false alarm may be expressed as
\beq
\label{Pfa}
P_{fa} = \Pr (T\geq \gamma | \mc{H}_0)
\eeq
and the probability of missed detection as
\beq
\label{Pmd}
P_{md} = \Pr (T<\gamma | \mc{H}_1)
\eeq
These probabilities depend on the distribution of $T$ under the two hypotheses. The probability distribution function (PDF) and the cumulative distribution function (CDF) of $T$ will be indicated as $f_{T|\mc H_i}(t)$ and $F_{T|\mc H_i}(t)$, respectively, for $i \in \{0,1\}$. Thus, (\ref{Pfa}) and (\ref{Pmd}) may be written as
\begin{eqnarray}
\label{Pfa1}
& P_{fa} = 1 - F_{T|\mc H_0}(\gamma) \\
\label{Pmd1}
& P_{md} = F_{T|\mc H_1}(\gamma)
\end{eqnarray}
In the next sections the distribution of $T$ in both cases will be derived, using tools from Random Matrix Theory (RMT) which allow to analyze the spectral properties of large-dimensional sample covariance matrices. This makes it possible to evaluate the detection performance, given a decision threshold, as well as to express the threshold as a function of the required probabilities of false alarm or missed detection (by inverting (\ref{Pfa1}) and (\ref{Pmd1})).

\section{False-alarm probability analysis}
\label{fa}

In this section, we first introduce some useful results from RMT that express the limiting distributions to which the largest and the smallest eigenvalues of $\bm{R}(N)$ converge as $N$ and $K$ grow.
Then, we exploit these theoretical results to find the limiting distribution of the test statistic $T$ and, through the relation (\ref{Pfa1}), we derive the false-alarm probability.

Most of the results of this section also appear, in a slightly different form, in \cite{commlett}. Here the results are stated in their entirety and are introduced by a a more rigorous mathematical derivation. Also, a new notation is adopted to emphasize the link between the Wishart case ($\mc H_0$) and the spiked-population case ($\mc H_1$, discussed in Sec. \ref{md}).

%

\subsection{Relevant results from Random Matrix Theory}

Under $\mc H_0$, since the columns of $\bm{Y}$ are zero-mean independent complex Gaussian vectors, the sample covariance matrix $\bm{R}(N)$ is a \textit{complex Wishart matrix} \cite{wishart}.

The fluctuations of the eigenvalues of Wishart matrices have been thoroughly investigated by RMT (see \cite{verdu} and \cite{bai} for an overview).
The most remarkable intuition of RMT
is that in many cases the eigenvalues of matrices with random entries turn out to converge to some fixed distribution,
when both the dimensions of the signal matrix tend to infinity with the same order.
For Wishart matrices the limiting joint eigenvalue distribution has been known for many years \cite{marchenko}; then, more recently, also the marginal distributions of single ordered eigenvalues have been found.

By exploiting some of these results, we are able to express the asymptotical values of the largest and the smallest eigenvalue of $\bm{R} (N)$ as well as their limiting distributions. We state the following theorem, which summarizes a number of relevant results.

\begin{theorem}
\label{theo1}
\emph{Convergence of the smallest and largest eigenvalues under $\mc H_0$.}
Let
\beq
\label{c}
c \triangleq \frac{K}{N}
\eeq
and assume that for $K,N \ra \infty$
\beq
c \ra \overline c \in (0,1)
\eeq
Define:
\begin{eqnarray}
& \mu_\pm(c) \triangleq \left( c^{1/2} \pm 1 \right)^2 \\
& \nu_\pm(c) \triangleq \left( c^{1/2} \pm 1 \right) \left( c^{-1/2} \pm 1\right)^{1/3}
\end{eqnarray}
Then, as $N, K \rightarrow \infty$, the following holds:

\begin{itemize}
\item[(i)] \textit{Almost sure convergence of the largest eigenvalue }\beq \hat{\lambda}_1 \convas \sigma^2_v \; \mu_+(c) \eeq
\item[(ii)] \textit{Convergence in distribution of the largest eigenvalue}
\beq
N^{2/3} \; \frac{\hat{\lambda}_1 - \sigma^2_v \; \mu_+(c)}{\sigma^2_v \; \nu_+(c)} \convd \mc W_2
\eeq
\item[(iii)] \textit{Almost sure convergence of the smallest eigenvalue} \beq \hat{\lambda}_K \convas \sigma^2_v \; \mu_-(c) \eeq
\item[(iv)]   \textit{Convergence in distribution of the smallest eigenvalue}
\beq
N^{2/3} \; \frac{\hat{\lambda}_K - \sigma^2_v \; \mu_-(c)}{\sigma^2_v \; \nu_-(c)} \convd \mc W_2
\eeq
where $\mc W_2$ is the Tracy-Widom law of order 2, defined in Appendix \ref{tw_app}.
\end{itemize}
\end{theorem}
\begin{proof}
The claims of this theorem follow from different results of RMT, up to some changes of variables and using a uniform notation. Proofs of the original theorems appear in the references listed below.\\
Claims (i) and (iii) descend from the work by Marchenko and Pastur \cite{marchenko}, later extended by Silverstein, Bai, Yin, \textit{et al.} \cite{bai}.\\
Claim (ii) was proved, under the assumption of Gaussian entries, by Johansson \cite{scarlett}, Johnstone \cite{john} and Soshnikov \cite{sosh}, and generalized to the non-Gaussian case by P\'ech\'e \cite{peche}. \\
Claim (iv) derives from a very recent result by Feldheim and Sodin \cite{small}.
\end{proof}

\subsection{Derivation of $F_{T|\mc H_0}$ and $P_{fa}$ }

The results of Theorem \ref{theo1} allow, through some algebraic manipulations, to determine the limiting distribution of the test statistic $T$ under the hypothesis $\mc H_0$.
Although the resulting distribution is obtained under the joint limit $K,N \rightarrow \infty$, simulations show that it provides an accurate estimation of the false-alarm probability already for not-so-large values of $K$ and $N$. Numerical results investigating this issue are presented in Sec. \ref{results}.

In order to apply claims (ii) and (iv), we define:
\begin{align}
& L_1 \triangleq N^{2/3} \; \frac{\hat{\lambda}_1 - \sigma^2_v \; \mu_+(c)}{\sigma^2_v \; \nu_+(c)}  \\
& L_K \triangleq N^{2/3} \; \frac{\hat{\lambda}_K - \sigma^2_v \; \mu_-(c)}{\sigma^2_v \; \nu_-(c)}
\end{align}
For the above-mentioned theorem, both $L_1$ and $L_K$ converge in distribution to the Tracy-Widom law $\mc W_2$:
\beq
f_{L_1}(z), f_{L_K}(z) \ra f_{\mc W_2}(z)
\eeq
where $f_{\mc{W}_2}(\cdot)$ represents the PDF associated with the law $\mc W_2$, as defined in Appendix \ref{tw_app}.

Then, from (\ref{T}), the test statistic $T$ becomes
\beq
T =\frac{\hat{\lambda}_1}{\hat{\lambda}_K} = \frac{ N^{-2/3} \nu_+(c) L_1 +  \mu_+(c) }{N^{-2/3}  \nu_-(c) L_K + \mu_-(c)}
\eeq
Notice that the term $\sigma^2_v$ is canceled out in the ratio (this is the reason that makes the detection threshold \vva blind'' with respect to the noise power).
We denote with $l_1$ and $l_K$, respectively, the numerator and the denominator of $T$, and with $\overline{f}_{l_1}(z)$ and $\overline{f}_{l_K}(z)$ their limiting PDFs for $N,K \ra \infty$. These distributions are the same as those of $L_1$ and $L_K$, up to a linear random variable transformation:
\begin{equation}
\overline{f}_{l_1}(z) = \frac{N^{2/3}}{ \nu_+(c)} f_{\mc{W}_2}\left( \frac{N^{2/3}}{ \nu_+(c)} (z-\mu_+(c)) \right)
\end{equation}
For the denominator, it must be observed that $\nu_-(c)<0$ for the considered range $c \in (0,1)$. Thus
\begin{align}
\overline{f}_{l_K}(z) & = \frac{N^{2/3}}{ |\nu_-(c)|} f_{\mc{W}_2}\left( \frac{N^{2/3}}{ |\nu_-(c)|} (\mu_-(c)-z) \right) \nonumber \\
& = -\frac{N^{2/3}}{ \nu_-(c)} f_{\mc{W}_2}\left( \frac{N^{2/3}}{ \nu_-(c)} (z-\mu_-(c)) \right)
\end{align}
To express the distribution of $T$, we assume that $f_{l_1}(l_1)$ and $f_{l_K}(l_K)$ are asymptotically independent, as it is reasonable for the size of the covariance matrix tending to infinity (and confirmed by following numerical results):
\beq
\label{indep}
\overline{f}_{l_1, l_K}(l_1, l_K) \approx  \overline{f}_{l_1}(l_1)  \overline{f}_{l_K}(l_K)
\eeq
Then, using the formula for the quotient of random variables \cite{ratiodist}, the resulting ratio distribution writes:
\begin{align}
\overline{f}_{T|\mc H_0}(t) & = \left[ \int_{-\infty}^{+\infty} |x| \overline{f}_{l_1,l_K}(tx,x) dx  \right] \cdot {I}_{\{t>1\}} \nonumber \\
& = \left[ \int_{0}^{+\infty} x \overline{f}_{l_1}(tx) \overline{f}_{l_K}(x) dx  \right] \cdot {I}_{\{t>1\}}
\label{pdfT}
\end{align}
where
the lower integration limit has been changed to $0$ instead of $-\infty$, since the covariance matrix is positive-semidefinite therefore all the eigenvalues are non-negative; the condition $t>1$ is necessary to preserve the order of the eigenvalues, since the distributions are defined under the assumption $l_1>l_K$.

Finally, we denote with $\overline{F}_{T|\mc H_0}(\gamma)$ the CDF corresponding to (\ref{pdfT}). For $N$ and $K$ large enough, we can approximate $F_{T|\mc H_0}(\gamma)$,
which is needed to compute $P_{fa}$ from (\ref{Pfa1}), with the asymptotical distribution:
\beq
F_{T|\mc H_0}(\gamma) \approx \overline{F}_{T|\mc H_0}(\gamma)
\eeq
The expression of $\overline{F}_{T|\mc H_0}$ depends on $N$ and $c$, i.e., $N$ and $K$. Simulation results show that the approximation is accurate for practical values of $N$ and $K$, also quite far from the asymptotical region.

Clearly, the practical interest in the relation between $P_{fa}$ and $\gamma$ found here is that it allows to determine the decision threshold as a function of the required false-alarm probability; this application is discussed in more detail in Sec. \ref{discussion}.

It is interesting to note that the distribution $F_{T|\mc H_0}$ for finite $N$ and $K$ can also be expressed \textit{exactly}, by following a completely different approach. This exact distribution and the corresponding detection threshold have been found in \cite{crowncom}. The drawback of the \vva exact'' approach is its complexity, which makes implementation difficult when $K$ and $N$ are large.

%
%
%
%
%
%
%
%
%
%
%
%
%
%
%
%
%


\section{Missed-detection probability analysis}
\label{md}

In this section we use an approach based on RMT to derive the limiting distribution of $T$ under $\mc H_1$ and consequently $P_{md}$.
As a preliminary step, we show that under this hypothesis $\bm R(N)$ can be reduced to a so-called \textit{spiked population model}, i.e., a model where the statistical covariance matrix is a finite-rank perturbation of the identity. Spiked population models were introduced by Johnstone \cite{john} and have an important role in Principal Component Analysis (PCA), with many statistical applications ranging from genetics to mathematical finance.
The fluctuations of the eigenvalues of sample covariance matrices constructed from spiked models are nowadays a hot topic in RMT.

\subsection{Reduction to the Spiked Population Model}

Under $\mc H_1$, the received signal matrix $\bm{Y}$ contains some Gaussian entries, like in the Wishart case, along with a certain number ($P$) of signal components.
In order to put into evidence the spiked structure of $\bm R(N)$, the received signal matrix $\bm Y$ (\ref{YN}) needs to be rewritten in the form
\beq
\bm Y = \bm{T Z}
\eeq
where $\bm T$ is a block matrix of size $K \times (P+K)$ defined as
\beq
\bm T = \left[
\begin{array}{c|c}
\frac{1}{\sigma_v} \bm {H \Sigma}^{1/2} & \bm{I}_K
\end{array}
\right]
\eeq
and $\bm Z$, of size $(P+K) \times N$, is defined as
\beq
\bm Z = \left[
\begin{array}{c}
\sigma_v \bm {\Sigma}^{-1/2} \bm{S} \vspace{1mm}\\
\hline
\vspace{-2mm}\\ \bm{V}
\end{array}
\right]
\eeq
This decomposition has been chosen such that all the entries $z_{ij}$ of $\bm Z$ ($ 1 \leq i \leq P+K, \; \; 1 \leq j \leq N$) have the following properties:
\begin{eqnarray}
\label{prop1}
& \mathrm{E}\; z_{ij} = 0 \\
& \mathrm{E}\; |z_{ij}|^2 = \sigma^2_v
\label{prop2}
\end{eqnarray}
which are necessary conditions for Theorem \ref{theo2} to hold.
The covariance matrix becomes then
\beq
\bm{R}(N) = \frac{1}{N} \bm{TZZ}^H\bm{T}^H
\eeq
which is exactly the model of \cite{baik_silv}, \cite{baik_benarous_peche} and \cite{feral_peche}.

Finally, we denote with $t_1, \hdots, t_K$ the eigenvalues of $\bm{TT}^H$. It follows from the structure of $\bm T$ that $P$ eigenvalues are different from $1$ (without loss of generality we put them in the first $P$ positions: $t_1 \geq \hdots \geq t_P$)  and the remaining $K-P$ are ones. To express the $P$ \vva spike eigenvalues'' (that represent the perturbation with respect to the pure-noise model), we notice that
\beq
\bm {TT}^H = \frac{1}{\sigma^2_v} \bm {H \Sigma H}^H + \bm I_K
\eeq
and the eigenvalues $t_1, \hdots, t_P$ result from the solution of
\beq
\label{poly2}
\begin{array}{l}
\mathrm{det}\left(\bm{H \Sigma H}^H - \sigma^2_v (t-1) \bm{I}_K \right)=0 \\
 \mbox{s.t. } t \neq 1
\end{array}
\eeq
The structure of the problem is identical to that of (\ref{poly1}), with the change of variable $s = \sigma^2_v (t-1)$. We can therefore conclude that the \vva spike eigenvalues'' $t_p$ are linked to the non-zero eigenvalues of the statistical covariance matrix, $s_p$, by the relation
\beq
\label{tp}
t_p = \frac{s_p}{\sigma^2_v}+1, \; \; \; 1 \leq p \leq P
\eeq
In general, the values of $s_p$ are calculated using (\ref{mdl}); in the case of \textit{single primary user} ($P=1$), there is the simplified expression (\ref{s1}) which leads to
\beq
t_1|_{P=1} = \| \bm{h} \|^2 \frac{\sigma^2}{\sigma^2_v}+1
\eeq
\subsubsection*{Relation between spike eigenvalues and SNR}
The spike eigenvalues are related with the SNR; this fact turns out to be useful especially in the case of $P=1$.
From (\ref{tp}) we can write
\beq
\sum_{p=1}^P t_p = \frac{1}{\sigma^2_v} \sum_{p=1}^P s_p + P
\eeq
but, from the eigendecomposition of $\bm{H \Sigma H}^H$ and from (\ref{snrP}) it follows that
\beq
\label{eigendiag}
\sum_{p=1}^P s_p = \mathrm{tr} \; \bm{H \Sigma H}^H = \rho K \sigma^2_v
\eeq
hence
\beq
\label{rho_eig}
\sum_{p=1}^P t_p = K \rho + P
\eeq
Therefore, in the \textit{case of one primary user} ($P=1$), the (unique) spike eigenvalue may be expressed directly as a function of the SNR:
\beq
\label{t1rho}
t_1|_{P=1} = K \rho +1
\eeq
Note that, by exploiting the property (\ref{eigendiag}), one could also obtain (\ref{s1}) without resorting to the characteristic equation.
%

In the \textit{case of multiple primary signals} ($P>1$), the \textit{sum} of the spike eigenvalues is related to the SNR, but not the \textit{single} eigenvalues. Therefore, to compute the $t_p$ (in particular $t_1$, which is needed to apply Theorem \ref{theo2}), it is necessary to know the channel matrix and the power of primary signals and use (\ref{sP}).

%
%
%
%

\subsection{Relevant results from Random Matrix Theory}
\label{obs}

We are now ready to state the following theorem which provides a useful result on the convergence of the largest eigenvalue in spiked population models.

\begin{theorem}
\label{theo2}
\emph{Convergence of the largest eigenvalue under $\mc H_1$.}
Again, assume that for $K,N \ra \infty$
\beq
c = \frac{K}{N} \ra \overline c \in (0,1)
\eeq
In addition, assume that for all $i,j$ s.t. $ 1 \leq i \leq P+K, \; \; 1 \leq j \leq N$:
\begin{itemize}
	\item[($A_1$)] $ \mathrm{E}\; z_{ij} = 0 $
	\item[($A_2$)] $\mathrm{E}\; (\Re\ z_{ij})^2 = \mathrm{E}\; (\Im\ z_{ij})^2 = \frac{\sigma^2_v}{2}$
	\item[($A_3$)] $\forall k>0, \; \; \mathrm{E}\; | z_{ij}|^{2k} < \infty$ and $\mathrm{E}\; (\Re \ z_{ij})^{2k+1} = \mathrm{E}\ (\Im\; z_{ij})^{2k+1} = 0$
	\item[($A_4$)] $\mathrm{E}\; (\Re \ z_{ij})^4 = \mathrm{E}\; (\Im \ z_{ij})^4  = \frac{3}{4}\sigma^4_v$
\end{itemize}
Define:
\begin{eqnarray}
& \mu_s(t_1, c) \triangleq t_1 \left( 1+ \frac{c}{t_1-1} \right) \\
& \nu_s(t_1, c) \triangleq t_1 \sqrt{1-\frac{c}{(t_1-1)^2}}
\end{eqnarray}
Then, as $N, K \rightarrow \infty$, the following holds:
\begin{itemize}
\item[(i)] \textit{Almost sure convergence of the largest eigenvalue: phase transition phenomenon}\\
If $t_1 > 1 + c^{1/2}$:
\beq
\hat{\lambda}_1 \convas \sigma^2_v \mu_s(t_1, c)
\eeq
If $t_1 \leq 1 + c^{1/2}$:
\beq \hat{\lambda}_1 \convas \sigma^2_v \; \mu_+(c) \eeq
\item[(ii)] \textit{Convergence in distribution of the largest eigenvalue} \\
Let $m$ (with $1 \leq m \leq P$) be the multiplicity of the first spike eigenvalue $t_1$.  \\
If $t_1 = \hdots = t_m > 1 + c^{1/2}$:
\beq
\label{Gm}
N^{1/2} \; \frac{\hat{\lambda}_1 - \sigma^2_v \; \mu_s(t_1,c)}{\sigma^2_v \; \nu_s(t_1,c)} \convd \mc G_m
\eeq
If $t_1 = \hdots = t_m = 1 + c^{1/2}$:
\beq
\label{Fm}
N^{2/3} \; \frac{\hat{\lambda}_1 - \sigma^2_v \; \mu_+(c)}{\sigma^2_v \; \nu_+(c)} \convd \mc{A}_m
\eeq
If $t_1 < 1 + c^{1/2}$:
\beq
\label{non_id}
N^{2/3} \; \frac{\hat{\lambda}_1 - \sigma^2_v \; \mu_+(c)}{\sigma^2_v \; \nu_+(c)} \convd \mc W_2
\eeq
where $\mc{A}_m$ and $G_m$ are distribution laws defined in Appendix \ref{at_app} and \ref{gue_app}, respectively.

\end{itemize}
\end{theorem}

\begin{proof}
The proof of claim (i) is due to Baik and Silverstein \cite{baik_silv}; claim (ii) was found by Baik, Ben Arous and P\'ech\'e \  \cite{baik_benarous_peche} under the additional assumption of $z_{ij}$ Gaussian with unit variance, and was generalized into this form by F\'eral, P\'ech\'e \cite{feral_peche} using results from Bai and Yao \cite{clt}.
\end{proof}

\subsection{Interpretation of the results}

\subsubsection{Validity of the assumptions}

All the assumptions ($A_1$)-($A_4$) are verified exactly for the noise part of $\bm Z$, whose entries are complex Gaussian random variables. For the signal part, the first two assumptions are guaranteed by construction of $\bm Z$: ($A_1$) is given by (\ref{prop1}) and ($A_2$) is equivalent to (\ref{prop2}) (provided that the variance of $\bm s$ is equally distributed between real and imaginary part, which is true for all types of complex signals used in communications). Assumption ($A_3$) is also verified in practical cases.

Assumption ($A_4$) is satisfied exactly by Gaussian signals, while there exist several types of signals (e.g. PSK, QAM) whose fourth moment is lower than that of a Gaussian random variable. However, since the type of primary signal is usually unknown to the secondary users, the Gaussian assumption is reasonable in general.
In addition, since $P<K$, most of matrix $\bm Z$ is represented by the noise part which does always satisfy ($A_4$):
therefore the theorem can be applied in almost all practical cases, even when this assumption does not hold exactly. The approximation introduced in this way is small and becomes negligible when the SNR of the primary signal is low, as shown in Sec. \ref{relax}.

\subsubsection{Phase transition phenomenon}

The first important result implied by the theorem is the existence of a \textit{critical value} of $t_1$ that determines whether a signal component is identifiable or not.
This behavior is called \textit{phase transition phenomenon}. 
In fact, when $t_1 \leq 1 + c^{1/2}$, the largest eigenvalue of the covariance matrix converges to the same value as in the pure-noise model, whereas for $t_1 > 1 + c^{1/2}$, it converges to a larger value: $\mu_s(t_1, c) > \mu_+(c)$. This property makes it possible to detect the presence of signals.

In case of $P=1$, the critical value can be expressed directly in terms of the SNR using (\ref{t1rho}):
\beq
\label{rho_crit}
\rho > \frac{1}{\sqrt{KN}}
\eeq
This relation also allows to determine the \textit{minimum number of samples} for the detector to be able to identify signals with a given SNR.

\subsubsection{Limiting distributions}

The second claim of the theorem clarifies \textit{how }the largest eigenvalue converges to the asymptotical limit.
For non-identifiable components, the limiting distribution is the same as in the case of no signal.
For components with eigenvalues placed exactly on the critical point, the limiting distribution is a generalization of the one encountered in the previous case: in fact, for $m=0$, $\mc{A}_0$ reduces to the Tracy-Widom law (Appendix \ref{at_app}).
For components above the critical value, we find the distributions $\mc G_m$: for $m=1$, which is the most common case in practical applications, $\mc G_1$ is simply the normal distribution; for $m=2$, we have derived a simple expression of the CDF of $\mc G_2$ in terms of the Gaussian error function (see Appendix \ref{gue_app}).

Finally, notice that both the events of eigenvalues exactly equal to the critical point and of eigenvalues with multiplicity larger than one are asymptotically \textit{events with zero probability}. The results concerning these cases are mentioned for completeness, but are not important for practical applications. Therefore, the case (\ref{Gm}) with $\mc G_1$ is by far the most important result of this theorem and allows to express $P_{md}$. Furthermore, $\mc G_1$ does not even involve complicated calculations because it reduces to the Gaussian distribution.

%
%
%


\subsection{Derivation of $F_{T|\mc H_1}$ and $P_{md}$}

Thanks to the results of Theorem \ref{theo2}, we are now able to express the limiting probability distribution of the test statistic $T$ under the hypothesis $\mc H_1$ and, consequently, to derive an analytical expression for the probability of missed detection. From now on, we refer to the case of \textit{identifiable signals}, i.e., we assume the $P$ signal components produce spiked eigenvalues above the critical limit $1+c^{1/2}$.

The approach that we adopt is the same as in the case of $\mc H_0$: we define again
\beq
L_1 \triangleq N^{1/2} \; \frac{\hat{\lambda}_1 - \sigma^2_v \; \mu_s(t_1,c)}{\sigma^2_v \; \nu_s(t_1,c)}
\eeq
which, for claim (ii), has a limiting PDF
\beq
f_{L_1}(z) \ra f_{\mc{G}_m}(z)
\eeq
where $f_{\mc{G}_m}(\cdot)$ represents the PDF associated with $\mc{G}_m$ ($m$ is the multiplicity of $t_1$), as defined in Appendix \ref{gue_app}.

As for the distribution of smallest eigenvalue,
we introduce the following theorem.

\begin{theorem}
\label{theo3}
\emph{Distribution of the $K-P$ smallest eigenvalues under $\mc H_1$.}
Assume that for $K,N \ra \infty$
\beq
c = \frac{K}{N} \ra \overline c \in (0,1)
\eeq
and that $t_p > 1+c^{1/2}$ for $1 \leq p \leq P$, the eigenvalues $\hat{\lambda}_{P+1} , \hdots, \hat{\lambda}_K$ of $\bm R(N)$ have asymptotically the same limiting distribution as those of a $(K-P) \times (K-P)$ Wishart matrix.
\end{theorem}
\begin{proof}
The result follows from the proof of Lemma 2 in \cite{kn}.
\end{proof}
Therefore, the distribution of the smallest eigenvalue is not affected by the presence of \vva spikes''
and claims (iii) and (iv) of Theorem \ref{theo1} can be applied also in this case with the only difference that, instead of $c$ (\ref{c}), now
\beq
c' = \frac{K-P}{N}
\eeq
Thus, we define
\beq
L_K \triangleq N^{2/3} \; \frac{\hat{\lambda}_K - \sigma^2_v \; \mu_-(c')}{\sigma^2_v \; \nu_-(c')}
\eeq
which still converges in distribution to the Tracy-Widom law
\beq
f_{L_K}(z) \ra f_{\mc W_2}(z)
\eeq

Then the test statistic $T$ becomes
\beq
T =\frac{\hat{\lambda}_1}{\hat{\lambda}_K} = \frac{ N^{-1/2} \nu_s(t_1,c) L_1 +  \mu_s(t_1,c) }{N^{-2/3}  \nu_-(c') L_K + \mu_-(c')}
\eeq
Also in this case the noise variance $\sigma^2_v$ is canceled out in the ratio. However, an implicit dependence on $\sigma^2_v$ remains in the term $t_1$, except for the case of single primary user ($P=1$) where $t_1$ is a function of the SNR only (\ref{t1rho}).

We denote with $l_1$ and $l_K$, respectively, the numerator and the denominator of $T$ and with $\overline{f}_{l_1}(z)$ and $\overline{f}_{l_K}(z)$ their limiting PDFs for $N, K \ra \infty$.
Through a random variable transformation, they may be expressed as
\begin{align}
\overline{f}_{l_1}(z) = \frac{N^{1/2}}{ \nu_s(t_1,c)} f_{\mc{G}_m}\left( \frac{N^{1/2}}{ \nu_s(t_1,c)} (z-\mu_s(t_1,c)) \right) \\
\overline{f}_{l_K}(z) = \frac{N^{2/3}}{ |\nu_-(c')|} f_{\mc{W}_2}\left( \frac{N^{2/3}}{ |\nu_-(c')|} (\mu_-(c')-z) \right)
\end{align}
Notice that, as a consequence of the observations in \ref{obs}, $\mc G_m$ is with probability one a Gaussian distribution
and thus it can be written in a more practical form as
\beq
\overline{f}_{l_1}(z) = \frac{(N/2\pi)^{1/2}}{ \nu_s(t_1,c)} \exp \left[ -\frac{N}{ 2 \, \nu^2_s(t_1,c)} \left(z-\mu_s(t_1,c)\right)^2 \right]
\eeq

Also in this      case, we assume $f_{l_1}(l_1)$ and $f_{l_K}(l_K)$ as asymptotically independent. The resulting limiting ratio distributions is
\begin{align}
\overline{f}_{T|\mc H_1}(t) & = \left[ \int_{-\infty}^{+\infty} |x| \overline{f}_{l_1,l_K}(tx,x) dx  \right] \cdot {I}_{\{t>1\}} \nonumber \\
& = \left[ \int_{0}^{+\infty} x \overline{f}_{l_1}(tx) \overline{f}_{l_K}(x) dx  \right] \cdot {I}_{\{t>1\}}
\label{pdfT1}
\end{align}
where, like in the previous case,
the domain of integration has been restricted to non-negative values, and the condition $t>1$ is necessary to ensure that $l_1>l_K$.

Finally, denoting with $\overline{F}_{T|\mc H_1}(\gamma)$ the CDF corresponding to the PDF in (\ref{pdfT1}),
we can take the approximation
\beq
F_{T|\mc H_1}(\gamma) \approx \overline{F}_{T|\mc H_1}(\gamma)
\eeq
that, in the asymptotical limit for $N$ and $K$, is the expression of the missed detection probability as it is given by (\ref{Pmd1}).
Numerical results show that the approximation is quite accurate for all cases of practical interest.

The relation between $P_{md}$ and $\gamma$ allows to predict the missed-detection probability of the detector with a given threshold, or to express the decision threshold as a function of the required probability of missed detection. The problem of setting the threshold is discussed in more detail in the next section.

\section{Setting the decision threshold}
\label{discussion}

The results presented in the previous sections express $P_{fa}$ and $P_{md}$ as a function of $\gamma$; therefore, by inverting the relations (\ref{Pfa1}) and (\ref{Pmd1}), the threshold can be expressed as a function of the error probabilities.

\subsection{Threshold as a function of $P_{fa}$}

The first relation
\beq
\label{gammapfa}
\gamma(P_{fa}) = F_{T|\mc H_0}^{-1} (1- P_{fa})
\eeq
allows to set the decision threshold accurately even if the noise power ($\sigma^2_v$) is unknown, since $F_{T|\mc H_0}$ depends only on the number of receivers ($K$) and of samples ($N$). The threshold set in this way, as a function of a target $P_{fa}$, is therefore a \vva blind' decision scheme as it is insensitive both to the noise and to the signal power.

In a previous work, Zeng and Liang \cite{liang} proposed a similar approach to set the decision threshold as a function of the probability of false alarm. Their detection algorithm was based on an approximated distribution of $T$, calculated taking into account only the limiting distribution of the largest eigenvalue (Theorem \ref{theo1}(ii)), and therefore provides non-optimal detection performance.
In \cite{debbah} another eigenvalue-based detection scheme was proposed, based only on the asymptotical values of $\hat{\lambda}_1$ and $\hat{\lambda}_K$ (Theorem \ref{theo1}(i)(iii)). For this reason, it does not allow to adjust the threshold as a function of $P_{fa}$ and is strongly sub-optimal with respect to our scheme unless $N$ and $K$ are extremely large.

A detailed performance comparison between the threshold based on the limiting distribution $ F_{T|\mc H_0}$ and these two previous approaches was provided in \cite{commlett}.

\subsection{Threshold as a function of $P_{md}$}

The second relation is
\beq
\gamma(P_{md}) = F_{T|\mc H_1}^{-1} (1- P_{md})
\eeq

Whereas $\gamma(P_{fa})$ has been found to depend only on $K$ and $N$, the expression of $\gamma(P_{md})$ depends also on the characteristics of the signal to be detected. In particular, two cases have to be considered separately:
\begin{itemize}
	\item when $P=1$, the only additional parameter needed to compute $P_{md}$ is the SNR $\rho$. In this case, the detector may still be defined \vva blind'' since it does not need to know explicitly the noise power nor the signal power. (Clearly, the detection performance has to be related, at least, with the SNR. For instance, in the case of Energy Detection, the SNR \textit{and} the noise power are needed to compute $P_{md}$.)
	\item when $P>1$, the knowledge of additional parameters is needed, namely the noise power ($\sigma^2_v$), the number of primary users ($P$), their powers ($\sigma^2_1, \hdots, \sigma^2_P$), and the channel ($\bm H$). These dependences arise from the nonlinear expression of $t_1$ (\ref{poly2}).
\end{itemize}

In general, all these parameters (even the SNR and the potential number of primary users) might be unknown.
Therefore, the relation between $\gamma$ and $P_{md}$ should better be used in the forward way, to predict the $P_{md}$ achieved using a given threshold under the possible primary signal scenarios, rather than to set the decision threshold according to a target $P_{md}$.
Nevertheless, if the system imposes a certain requirement on $P_{md}$ to keep the interference caused by the secondary network below a maximum level, the formula is useful to determine $\gamma$ based on the worst-case scenario (i.e., the one with the highest missed-detection probability) so as to guarantee in all cases the required protection to the primary network.

\subsection{Complexity and practical implementation}

As shown in \cite{liang} and \cite{commlett}, eigenvalue-based detection schemes offer a substantial performance improvement compared to ED (and a complete protection to noise uncertainty) at the price of an increased complexity.
Most of the computational complexity of these algorithms derives from the computation of the covariance matrix and of its eigenvalues: in \cite{liang} it is estimated that such operations lead to a complexity that grows as $K^3$, whereas in the case of ED it grows linearly with $K$. This increased computational cost is not dramatic, since the number of receivers is never enormous. On the other hand, in terms of the sample number (which is, actually, very large) the complexity remains linear with $N$ for both EBD and ED.

However, it is important to remark that the computational complexity is not influenced by the computation of the threshold. Even if the formulae found in this paper to express the threshold are very complex, they are always implemented off-line, and what the detector uses is simply a look-up table (LUT) containing several values of $\gamma$ as a function of $N$, $K$ $P_{fa}$, and/or $P_{md}$ and SNR. The use of LUTs also allows to change the decision threshold \vva on the fly'', in case of modifications of the system requirements.

Finally, for the computation of the distribution functions defined in this paper, routines are available on the web (e.g., \cite{tw_num} for the Tracy-Widom distributions) or can be implemented directly from the definitions given in the Appendices.

\section{Numerical results}
\label{results}

In this section, the results derived analytically in the previous sections are validated by comparing them with empirical results, obtained from Matlab Monte-Carlo simulations. The parameters used in the simulations are described in each sub-section; when referring to the SNR, it is defined according to (\ref{snr}).

\begin{figure}[p]
\hspace{-2mm}
	\centering
		\includegraphics[width=0.49\textwidth]{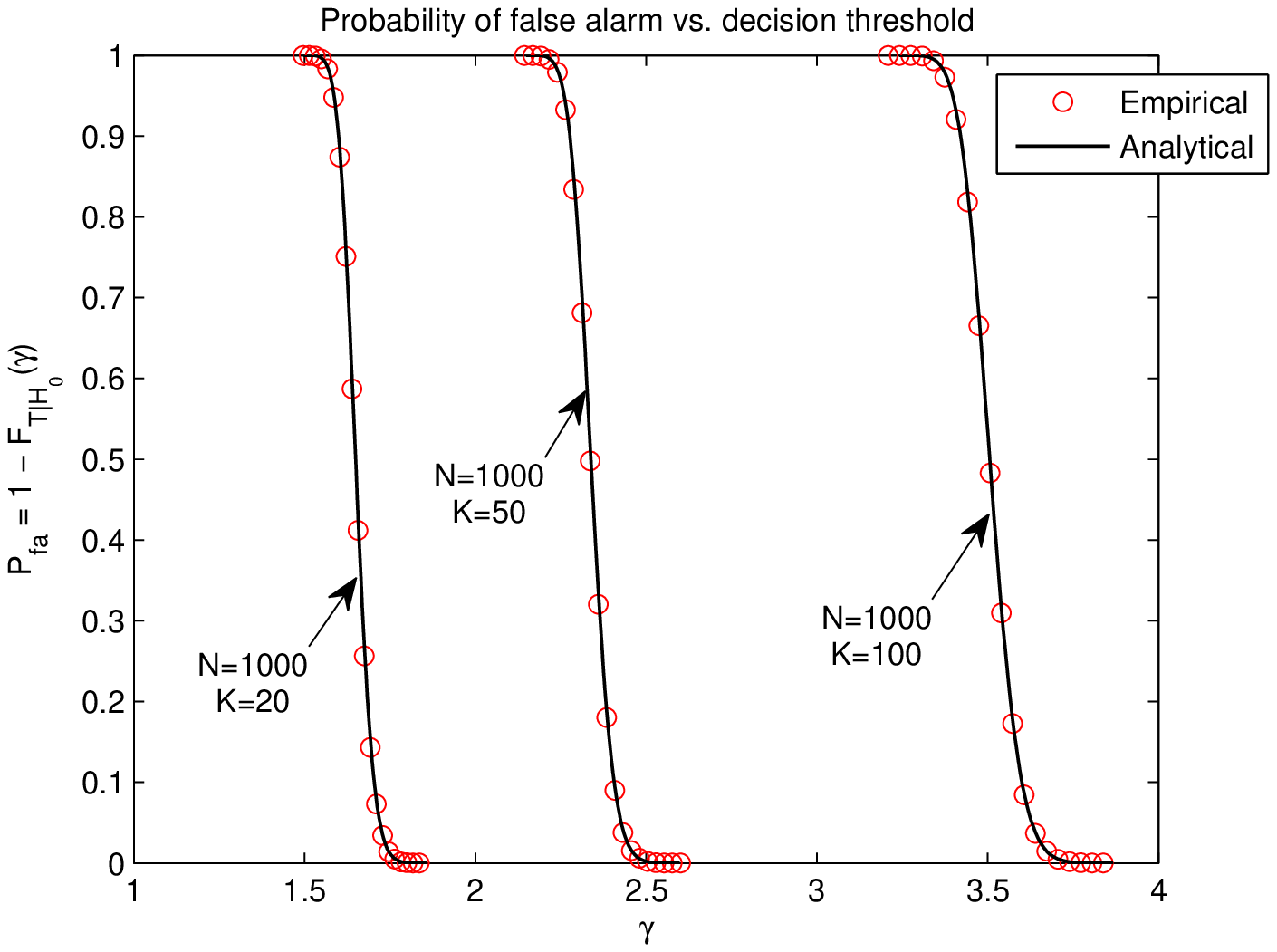}
	\caption{False-alarm probability: empirical vs. analytical.}
	\label{fth0}
\end{figure}

\begin{figure}[p]
\hspace{-2mm}
	\centering
		\includegraphics[width=0.49\textwidth]{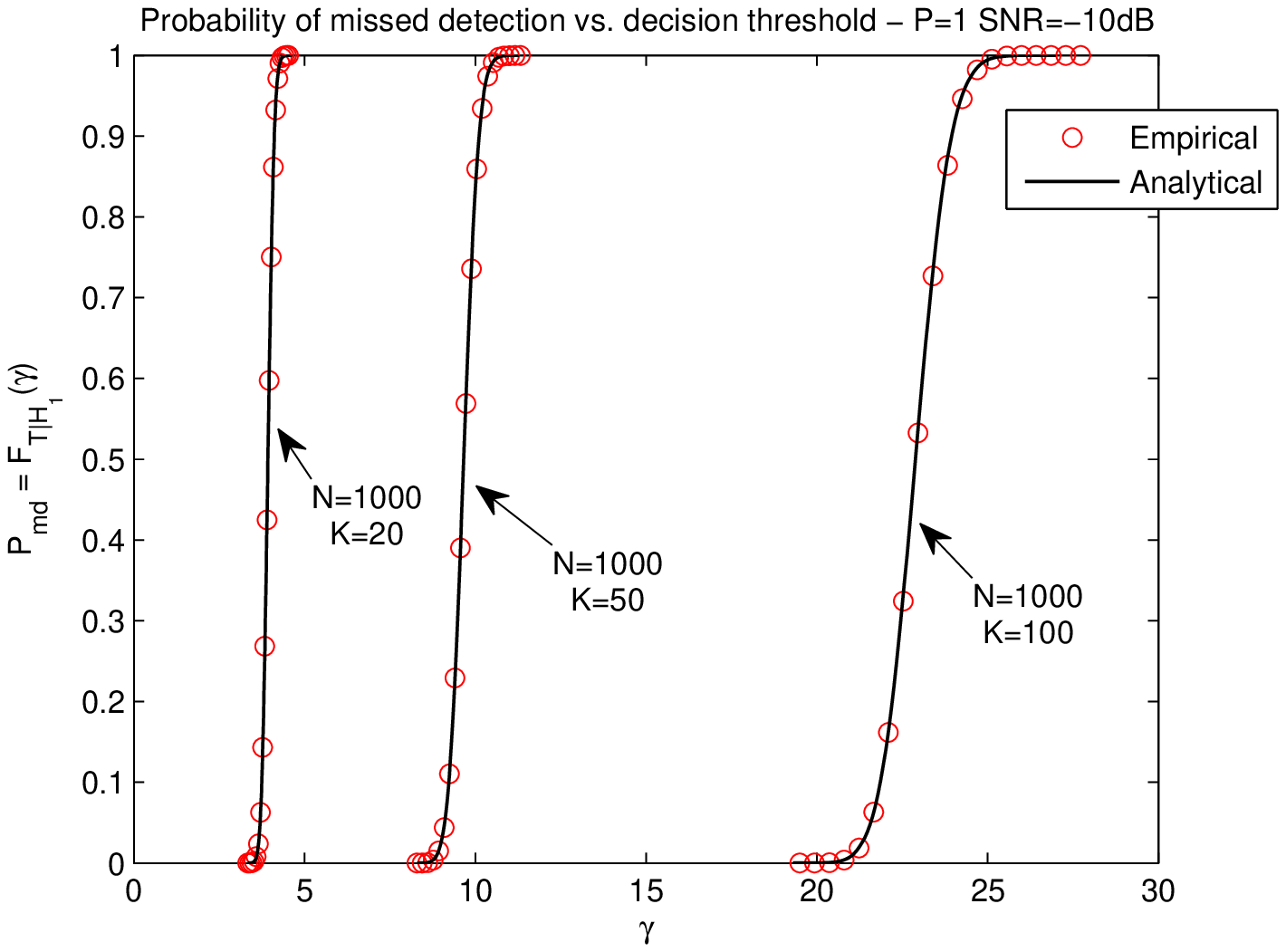}
	\caption{Missed-detection probability: empirical vs. analytical. $P=1$, $\rho=-10$dB.}
	\label{fth1p1snr-10}
\end{figure}

\begin{figure}[p]
\hspace{-2mm}
	\centering
		\includegraphics[width=0.49\textwidth]{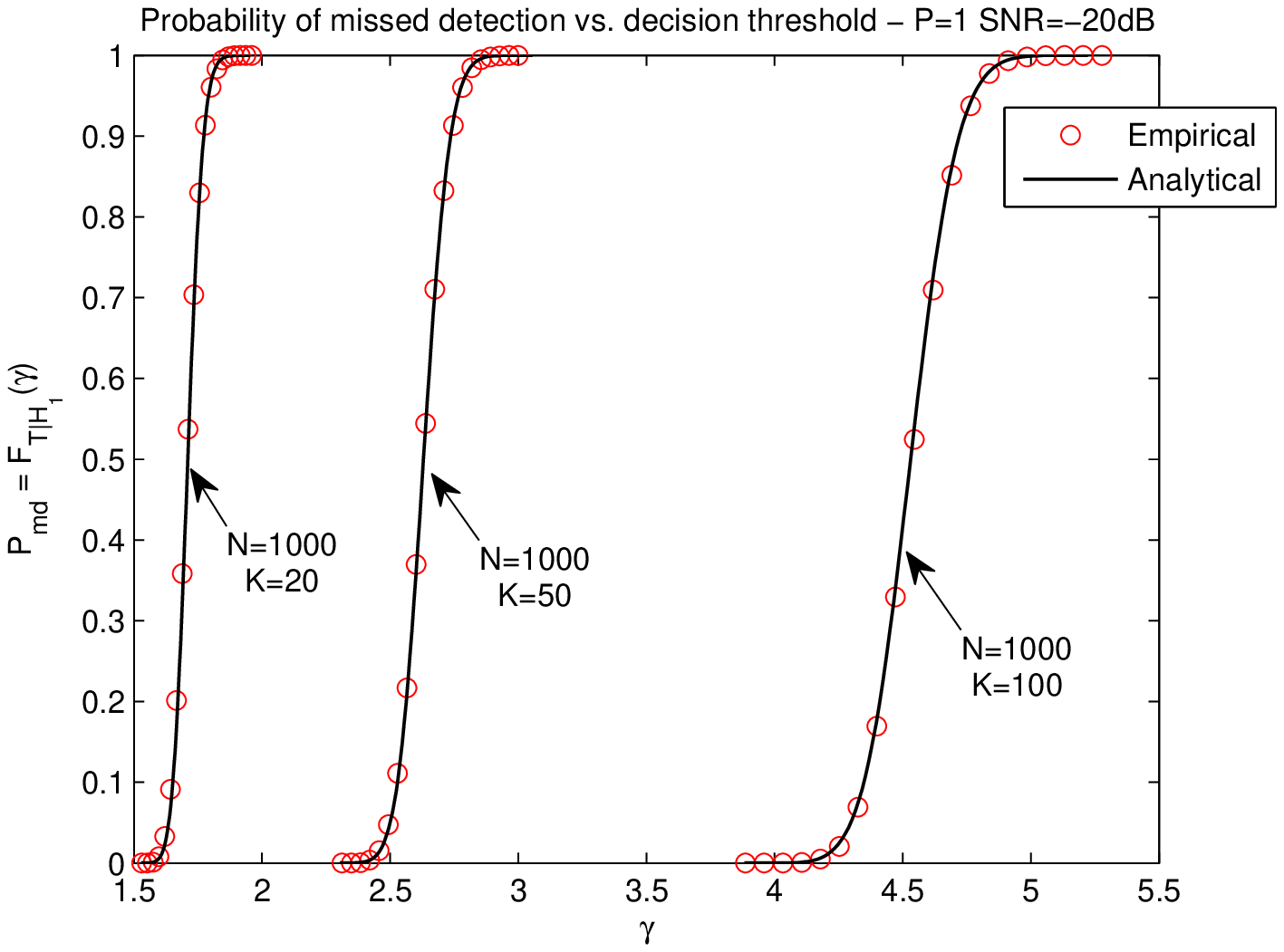}
	\caption{Missed-detection probability: empirical vs. analytical. $P=1$, $\rho=-20$dB.}
	\label{fth1p1snr-20}
\end{figure}

\begin{figure}[p]
\hspace{-2mm}
	\centering
		\includegraphics[width=0.49\textwidth]{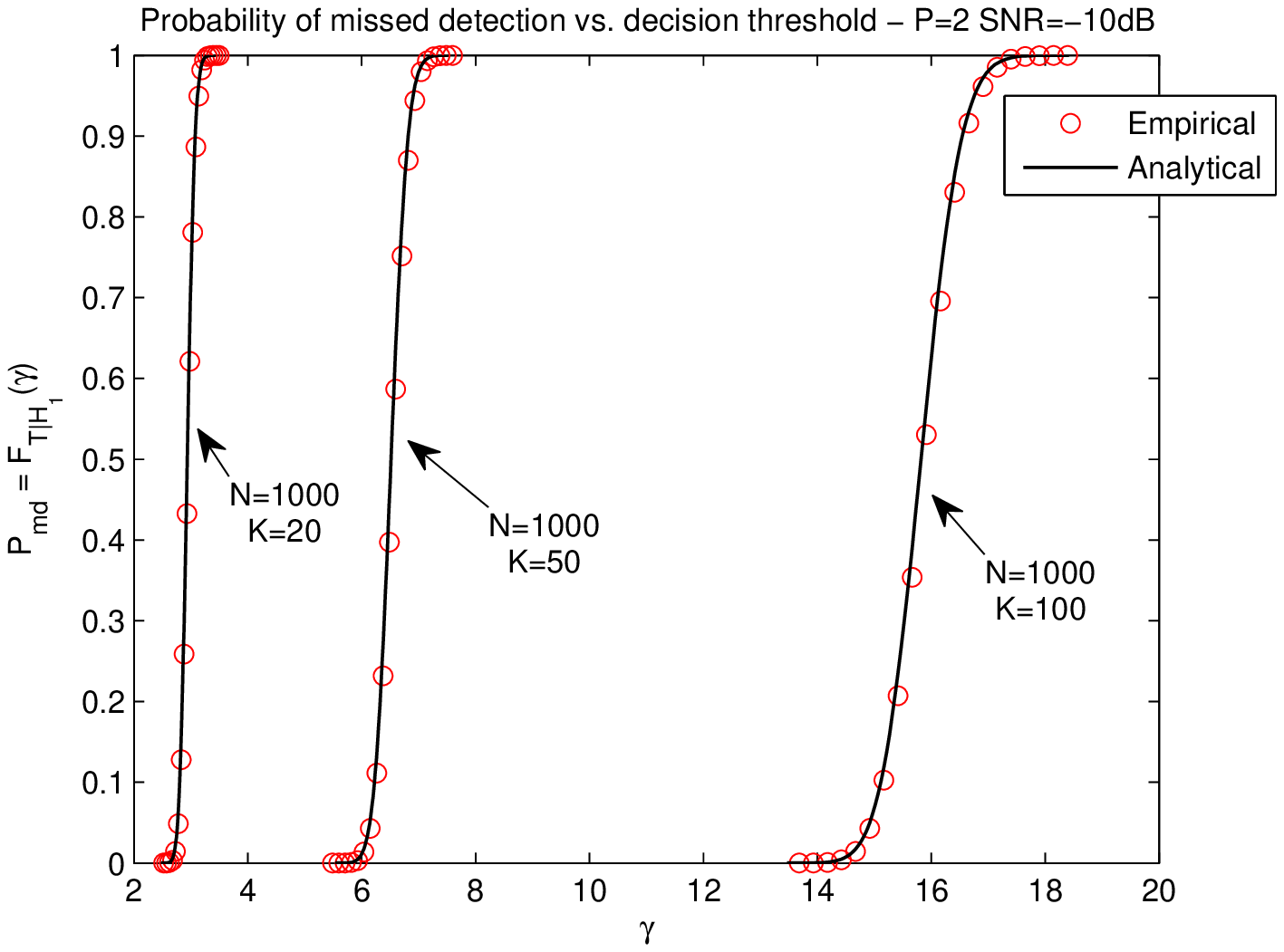}
	\caption{Missed-detection probability: empirical vs. analytical. $P=2$, $\rho=-10$ dB ($\rho_1 = 0.06 \approx -12.2$ dB, $\rho_2 = 0.04 \approx -14.0$dB).}
	\label{fth1p2snr-10}
\end{figure}

\subsection{Distribution of $T$ under $\mc H_0$}
Figure \ref{fth0} represents the probability of false alarm, i.e., the complementary CDF of $T$ under $\mc H_0$, for $N=1000$ and different values of $K$ (i.e., of $c$). The value of $\sigma^2_v$ has no effect, as it gets canceled out in the test statistic.

The curve predicted using the analytical expression turns out to be consistent with the empirical data in all the considered cases.
Comparing the three curves obtained with different values of $K$, one may observe that for a given $\gamma$ the probability of false alarm increases with $K$. However, this does not mean that the detector performance worsens for larger $K$, because also the curve of $P_{md}$ shifts rightwards, and consequently the decision threshold. The overall effect is indeed an improvement of performance when $K$ gets larger, as expected intuitively.

\subsection{Distribution of $T$ under $\mc H_1$}
Figures \ref{fth1p1snr-10}, \ref{fth1p1snr-20} and \ref{fth1p2snr-10} show the probability of missed detection, i.e., the CDF of $T$ under $\mc H_1$, for the same values of $N$ and $K$ as in the previous case.

The entries of $\bm{H}$ are taken as zero-mean complex Gaussian random coefficients (Rayleigh fading), with a variance normalized so as to obtain the desired SNR.
In the first figure the SNR is $-10$ dB with $P=1$ primary signal; in the second one, the SNR is $-20$ dB again with $P=1$; in the third one, $P=2$ with a global SNR of $-10$ dB (from (\ref{snrP}) with: $\rho_1 = \frac{\sigma^2_1 \| \bm{h}_1 \|^2}{K \sigma^2_v} = 0.06 \approx -12.2$ dB; $\rho_2 = \frac{\sigma^2_2 \| \bm{h}_2 \|^2}{K \sigma^2_v} = 0.04 \approx -14.0$ dB; $\sigma^2_v = 1$).
Notice that in the last case ($P>1$) the largest spike eigenvalue $t_1$, which determines $P_{md}$, depends on all the entries of $\bm H$ and not only on the SNR. In our simulations $t_1 = \{2.25, 4.04, 7.60 \}$, respectively for $K= \{20, 50, 100\}$.

Also in this case, the analytical curves fit the empirical data well in all the considered cases. We have considered low values of SNR, since the low-SNR region is the most important both from the theoretical point of view ($t_1$ close to the critical value of identifiability) and from the practical point of view (the challenge for cognitive radios is to detect signals also in presence of fading or shadowing).

As previously mentioned, the curves of $P_{md}$ shift rightwards as $K$ increases, i.e., the missed-detection probability gets lower for a given $\gamma$. This fact compensates the increase of $P_{fa}$ resulting in a larger separation between $1-F_{T|\mc H_0}$ and $F_{T|\mc H_1}$ for larger $K$.

\begin{figure}[p]
\hspace{-2mm}
	\centering
		\includegraphics[width=0.49\textwidth]{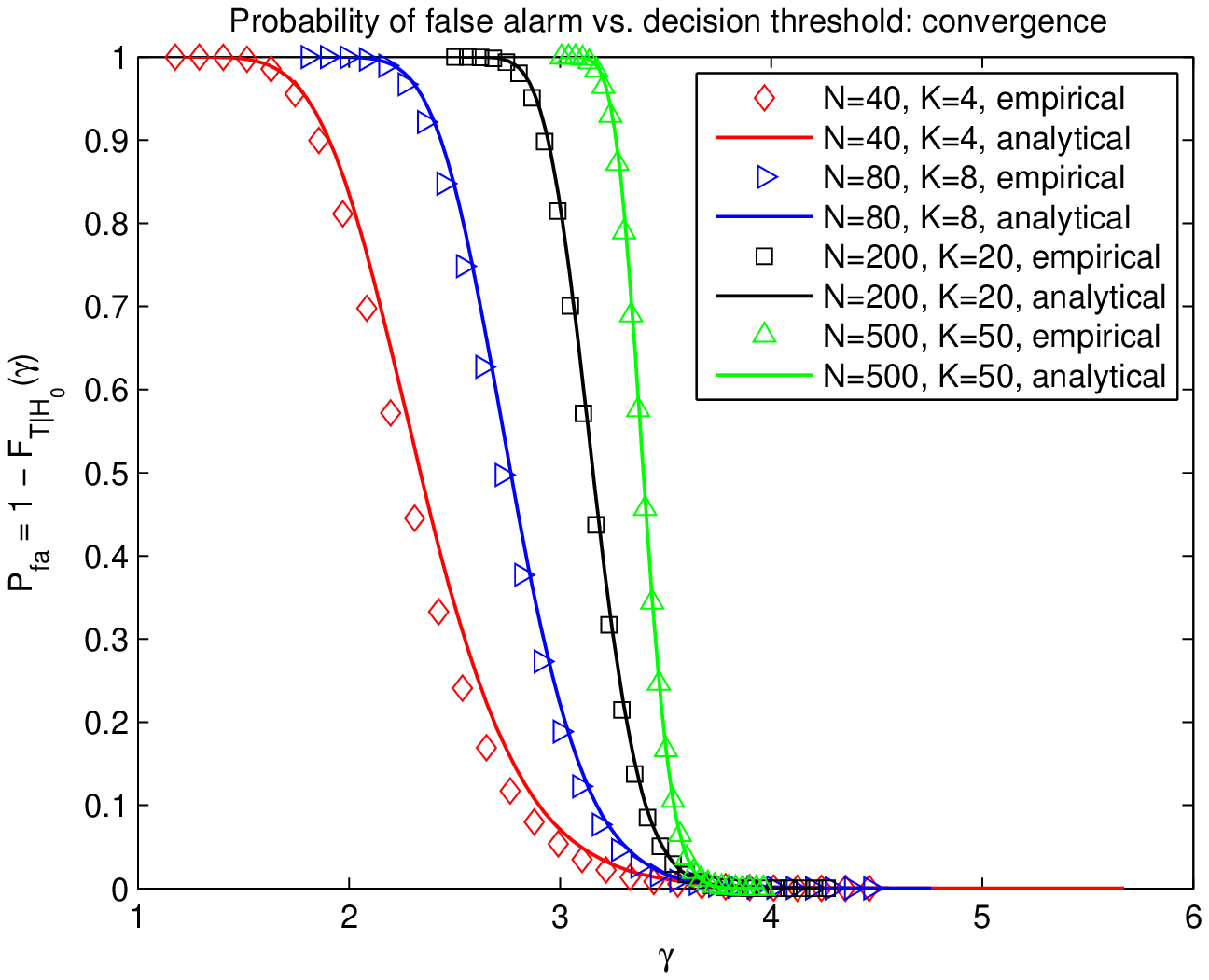}
	\caption{False-alarm probability: convergence, for a fixed $c=K/N=0.1$.}
	\label{convh0}
\end{figure}

\begin{figure}[p]
\hspace{-2mm}
	\centering
		\includegraphics[width=0.49\textwidth]{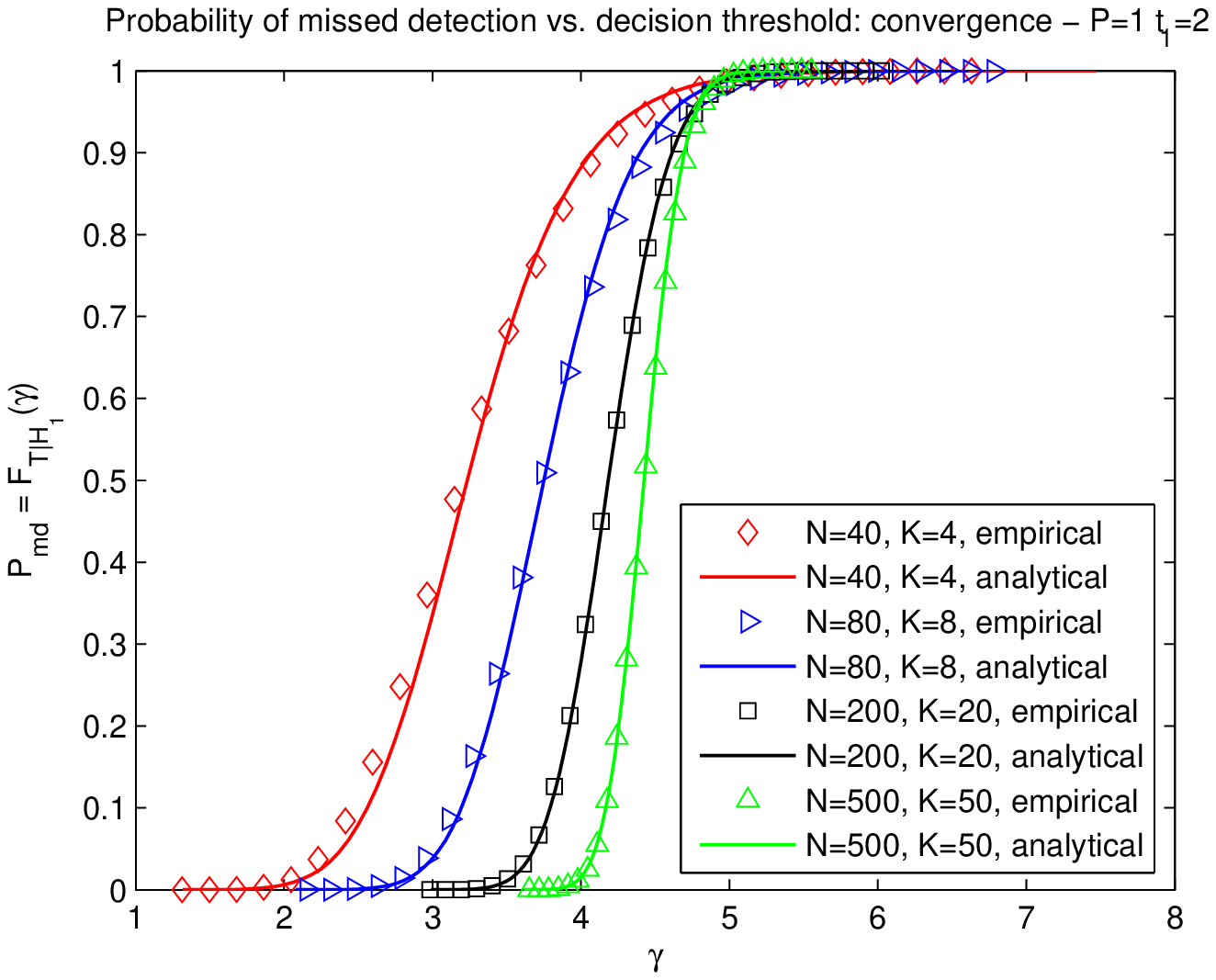}
	\caption{Missed-detection probability: convergence, for a fixed $c=K/N=0.1$. $P=1$, $t_1 = 2$.}
	\label{convh1}
\end{figure}

\subsection{Convergence}

Figures \ref{convh0} and \ref{convh1} show the convergence of the empirical CDFs to the analytical CDFs, which are calculated under asymptotical assumptions for $N$ and $K$. Four different couples of $\{N,K\}$ have been considered while keeping their ratio $c$ fixed at $0.1$.
Remarkably, even though the CDFs are asymptotical they provide an accurate approximation of the empirical CDFs also for low $K$ and $N$.

In the case $\mc H_0$, as $N$ and $K$ increase the CDF tends to a step function, because the largest and the smallest eigenvalues converge (almost surely) to the values $\mu_+(c)$ and $\mu_-(c)$, respectively; the variance instead depends also on $N$ (it gets smaller for larger $N$).

For the case $\mc H_1$, we considered a scenario with $P=1$ and, to make the comparison more evident, we kept $t_1$ fixed instead of the SNR ($\rho$ and $t_1$ are linked by a factor $K$, so they can not remain both constant with different $K$). In particular we chose the value $t_1=2$, which is above the critical value that is $1+\sqrt{c}=1.3162$ for all the considered couples of $\{N,K\}$. Similarly as in the previous case, the CDFs turn out to converge to a step function corresponding to the almost sure asymptotical limits of the eigenvalues.

\begin{figure}[p]
\hspace{-2mm}
	\centering
		\includegraphics[width=0.49\textwidth]{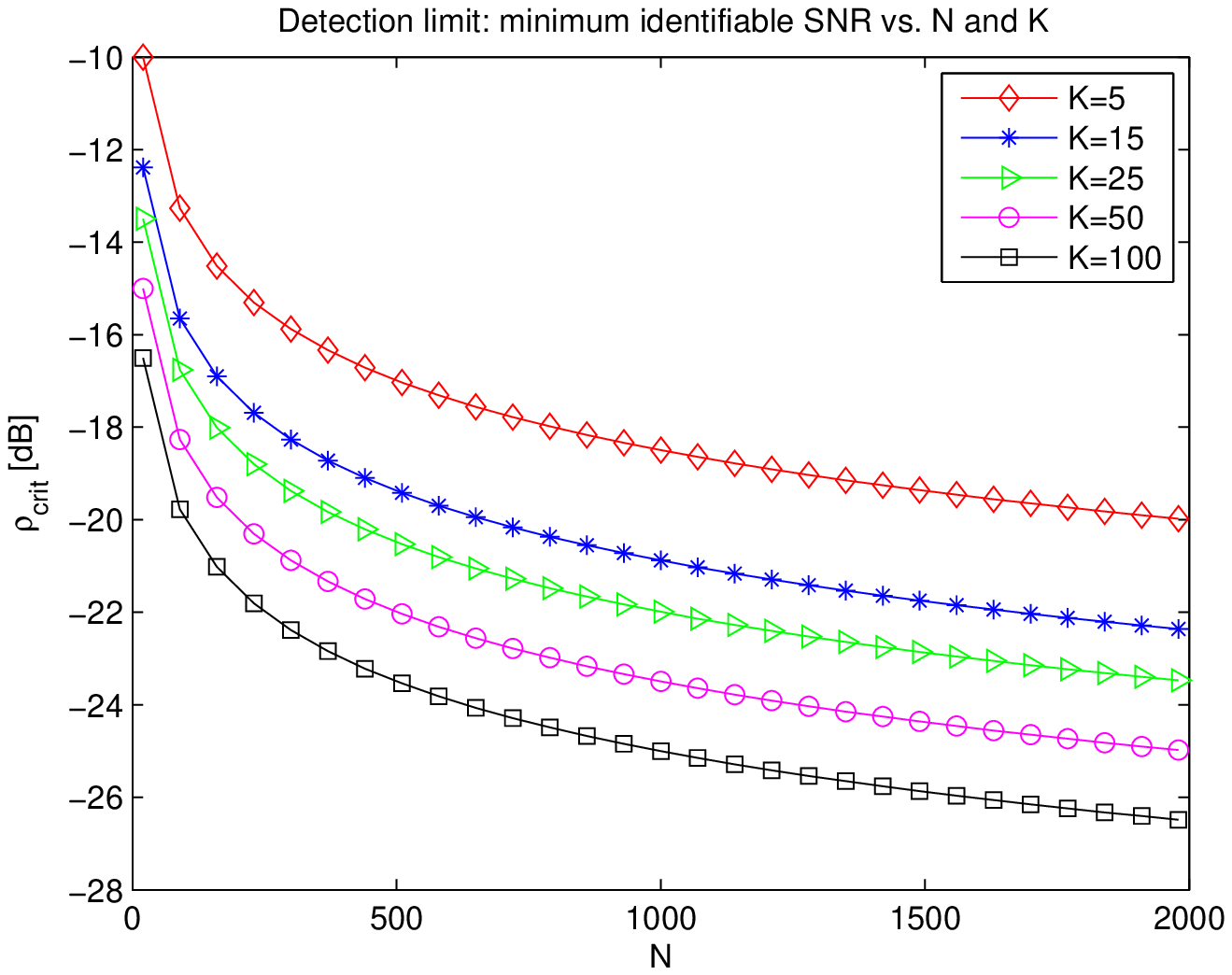}
	\caption{Phase transition phenomenon: minimum identifiable SNR vs. $N$ and $K$.}
	\label{ident}
\end{figure}

\begin{figure}[p]
\hspace{-2mm}
	\centering
		\includegraphics[width=0.49\textwidth]{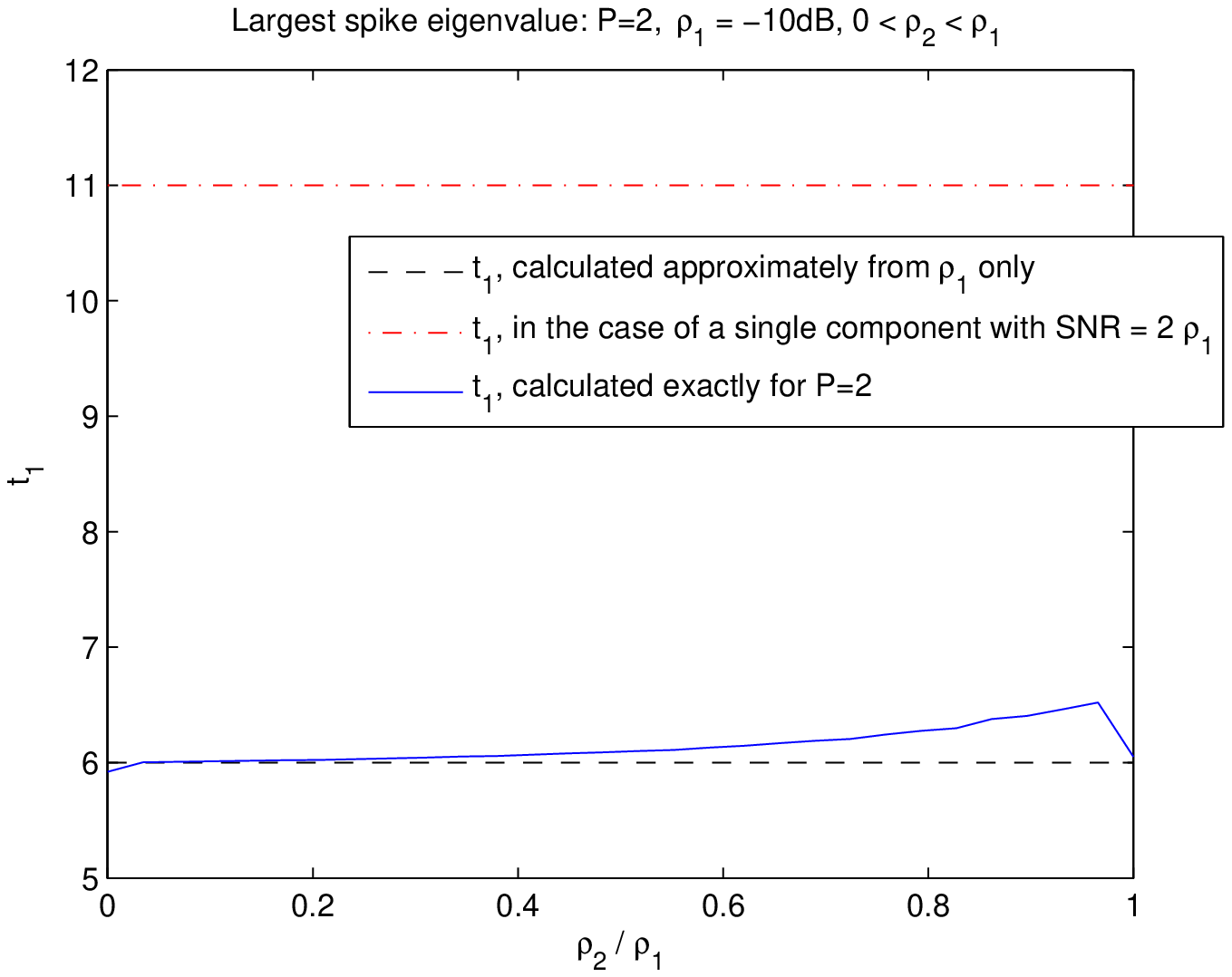}
	\caption{Impact of the approximated SNR formula ($\rho \approx \rho_1$) on the calculation of the largest spike eigenvalue ($t_1$).}
	\label{nuova}
\end{figure}

\subsection{Identifiability}
As a result of the phase transition phenomenon of Theorem \ref{theo2}, signals below a certain power level are not identifiable. A detection limit as a function of the SNR is expressed by the relation (\ref{rho_crit}), valid for $P=1$.
Figure \ref{ident} represents graphically the critical SNR for detection as a function of the number of samples $N$ and of receivers $K$. The relation may be used to determine the minimum sensing duration (i.e., the minimum number of samples) needed to detect signals for a required detector sensitivity.

A relation between identifiability threshold and SNR is valid only for $P=1$. For multiple signals, the expression of $t_1$ is more complex and does not depend only on the SNR.
However, it turns out that also for $P>1$ the value of $t_1$ is determined essentially by the power of the largest signal, i.e., by the SNR as if the first component was alone.
Therefore, we may define an approximated expression of the SNR, similar as (\ref{snr1}), depending only on the power of the dominant signal component:
\beq
\label{rho_app}
\rho \approx \frac{\max_p \left( \sigma^2_p \| \bm h_p\|^2 \right) }{K \sigma^2_v}
\eeq
This expression can be used in (\ref{rho_crit}) to determine, approximately, the parameters $N$ and $K$ of the detector.
As an example, in figure \ref{nuova} we consider the case $P=2$ with $\rho_1$ fixed at $0.1 = -10$ dB and $\rho_2$ varying from $0$ and $\rho_1$. The graph shows $t_1$ as a function of $\rho_2$, comparing the case when $t_1$ is calculated from the exact formula for $P=2$ (\ref{poly1}) with the case when it is calculated taking into account the largest component only (\ref{rho_app}) and with the case of a single component, but with double power (SNR $= 2 \rho_1$). It turns out that the actual value of $t_1$ is very close to the approximated one, even when the sum of $\rho_1$ and $\rho_2$ is close to $2 \rho_1$. Furthermore, the approximated $t_1$ tends to underestimate the actual $t_1$, resulting in a conservative choice of $N$ and $K$.

\begin{figure}[p]
\hspace{-2mm}
	\centering
		\includegraphics[width=0.49\textwidth]{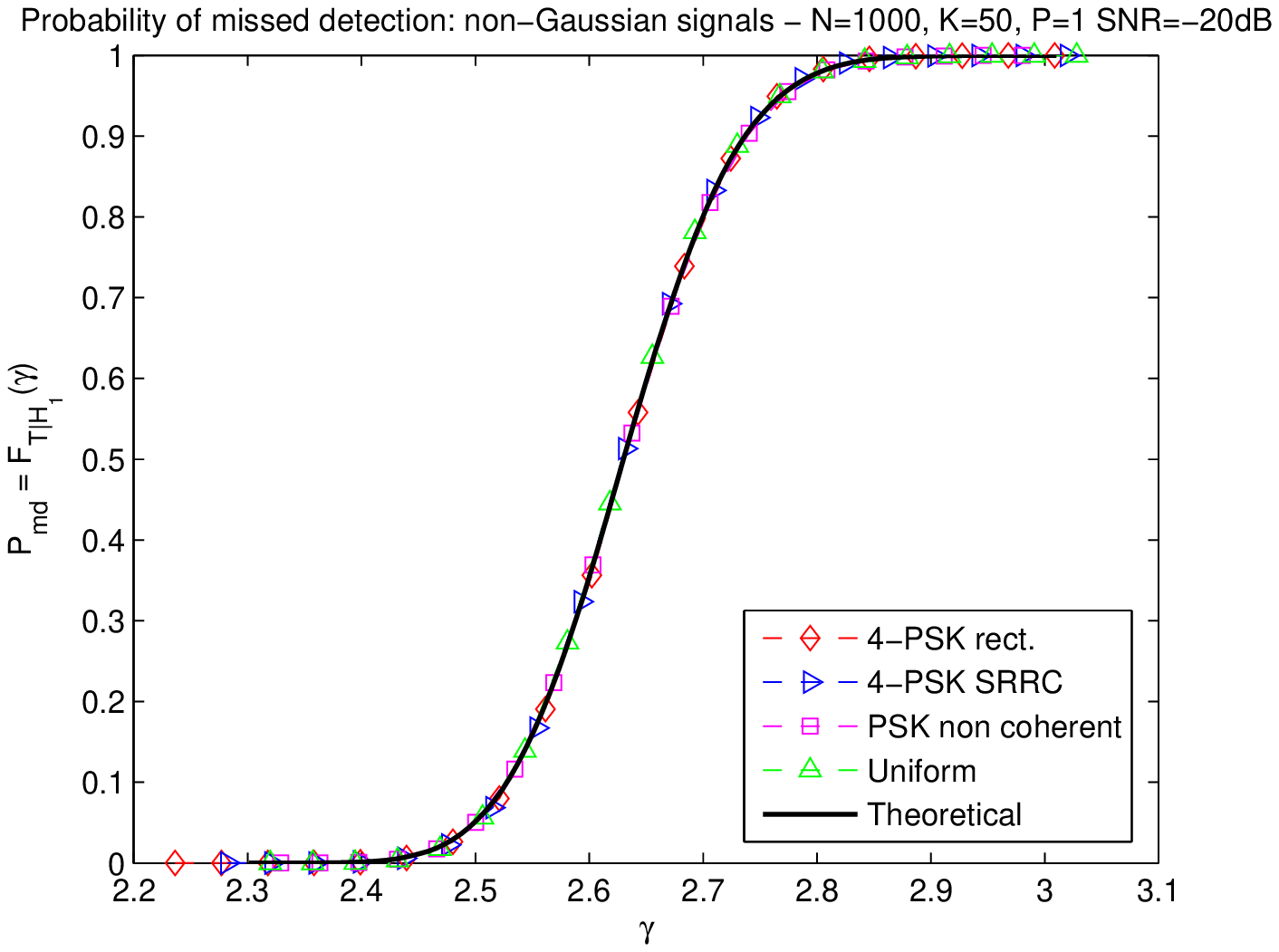}
	\caption{Missed-detection probability: empirical vs. analytical with non-Gaussian signals.. $P=1$, $\rho=-20$dB.}
	\label{types-20}
\end{figure}

\begin{figure}[p]
\hspace{-2mm}
	\centering
		\includegraphics[width=0.49\textwidth]{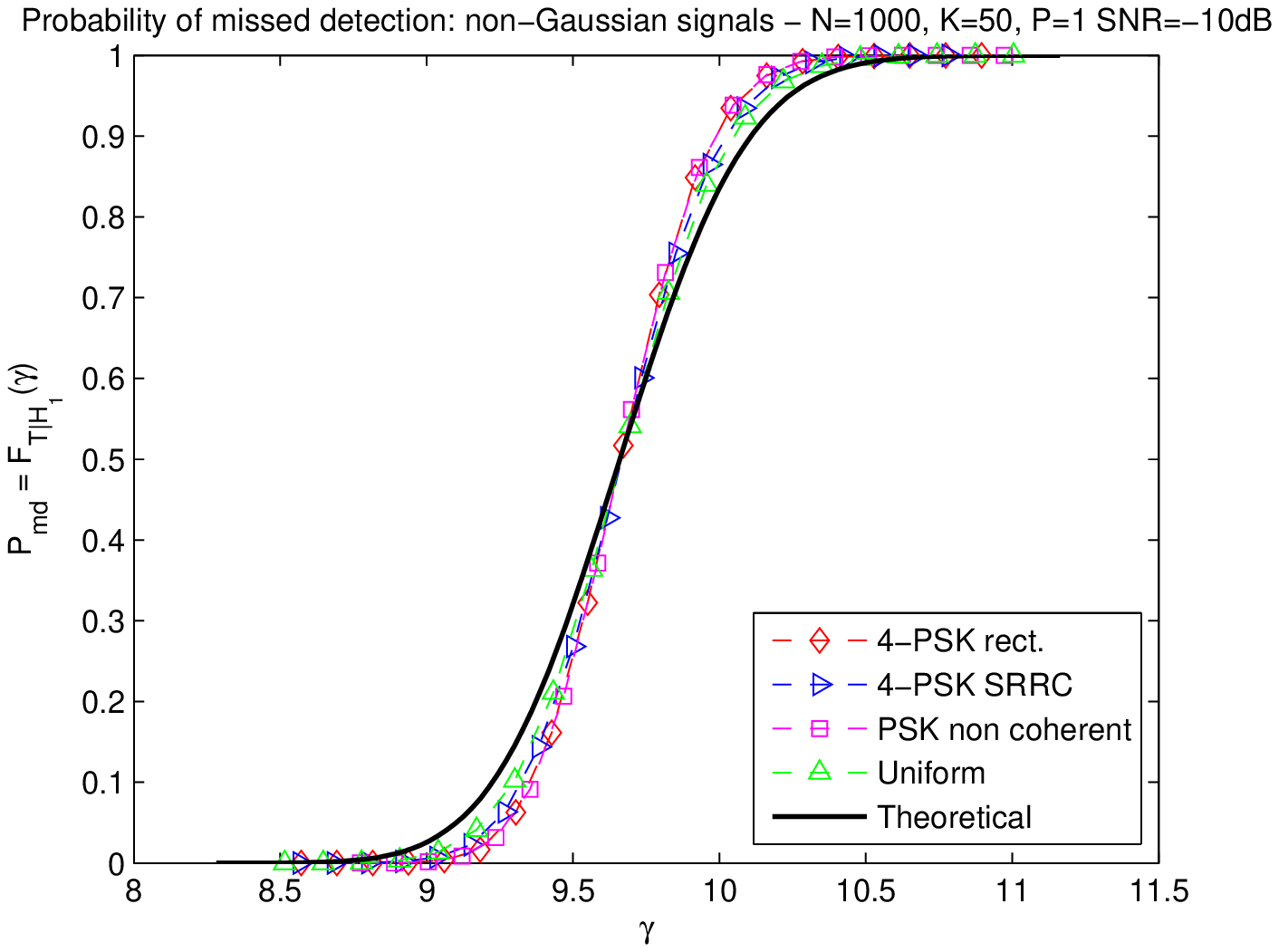}
	\caption{Missed-detection probability: empirical vs. analytical with non-Gaussian signals.. $P=1$, $\rho=-10$dB.}
	\label{types-10}
\end{figure}

\subsection{Non-Gaussian signals}
\label{relax}
As pointed out in Sec. \ref{obs}, the last assumption of Thereom \ref{theo2} is often not satisfied in practice, since realistic signals have typically a fourth moment lower than that of a Gaussian random variable.
Figures \ref{types-20} and \ref{types-10} show how the theoretical results, which rely on that assumption, fit empirical data obtained using more realistic types of primary signal. We considered four different types of signals, all with the same variance as in the Gaussian case, but with different fourth moments. The first curve refers to a 4-PSK modulated primary signal, with ideal rectangular pulse-shape filter and assuming a coherent reception; in the second curve, the signal is the same but passed through a square root raised cosine (SRRC) filter with roll-off $\alpha=0.5$; the third curve is a PSK signal with non-coherent reception (i.e., each sample has a random phase); the last curve refers to a random complex signal whose real and imaginary parts are uniformly distributed.

In the first figure, when the SNR is very low ($-20$ dB), the theoretical distribution fits the empirical data perfectly in spite of the fourth moment of the signals. When the SNR increases ($-10$ dB), some difference between the theoretical and the empirical curve can be observed, especially for PSK signals. It is interesting to notice that the Gaussian approximation on the fourth moment \textit{affects the variance of the resulting distribution, but not the mean}. The result is that the analytical formula overestimates the probability of missed detection (the interesting part of the curve is for $P_{md}<0.5$, i.e., the left tail).

To obtain a more accurate estimation of the missed-detection probability in case of non-Gaussian signals, for high SNR, one should add a \vva correction coefficient'' to the theoretical variance $\nu_s(t_1,c)$. Such coefficients would depend on the fourth moment of the signals, $\sigma^4_p$, and would be therefore specific of the modulation used.
It might be possible to determine by simulation the correction coefficients for a particular signal as a function of the SNR, whereas determining them analytically is a more challenging task since the matrix $\bm Z$ is composed of heterogeneous entries. However, the Gaussian assumption is valid asymptotically for $\rho \ra -\infty$ (the signal part in $\bm Z$ becomes negligible) and is accurate enough in the low-SNR region as shown by figure \ref{types-20}.

\begin{figure}[btp]
\hspace{-2mm}
	\centering
		\includegraphics[width=0.49\textwidth]{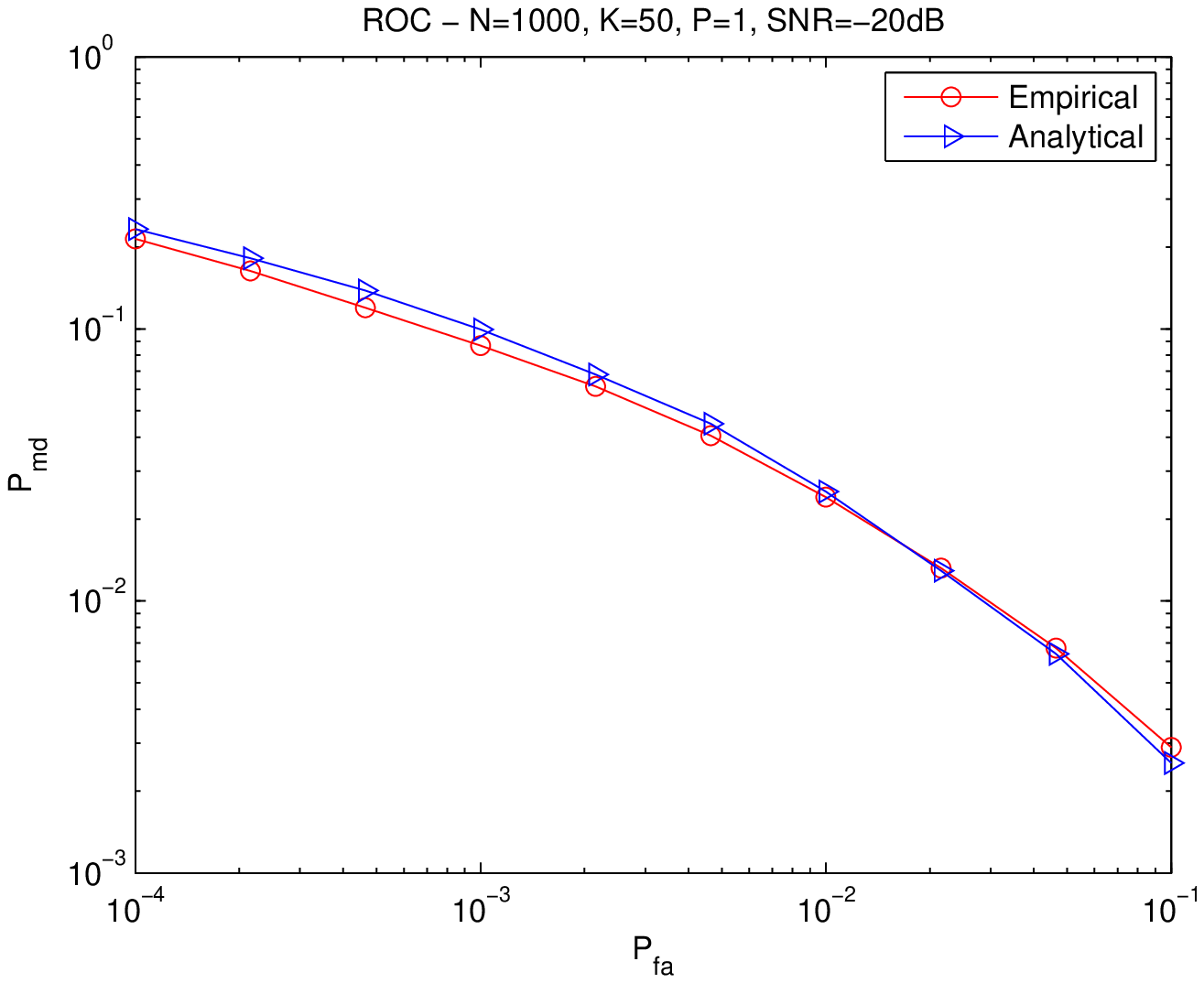}
	\caption{Complementary ROC: analytical vs. empirical. $P=1$, $\rho = -20$ dB.}
	\label{rocP1-20}
\end{figure}

\begin{figure}[btp]
\hspace{-2mm}
	\centering
		\includegraphics[width=0.49\textwidth]{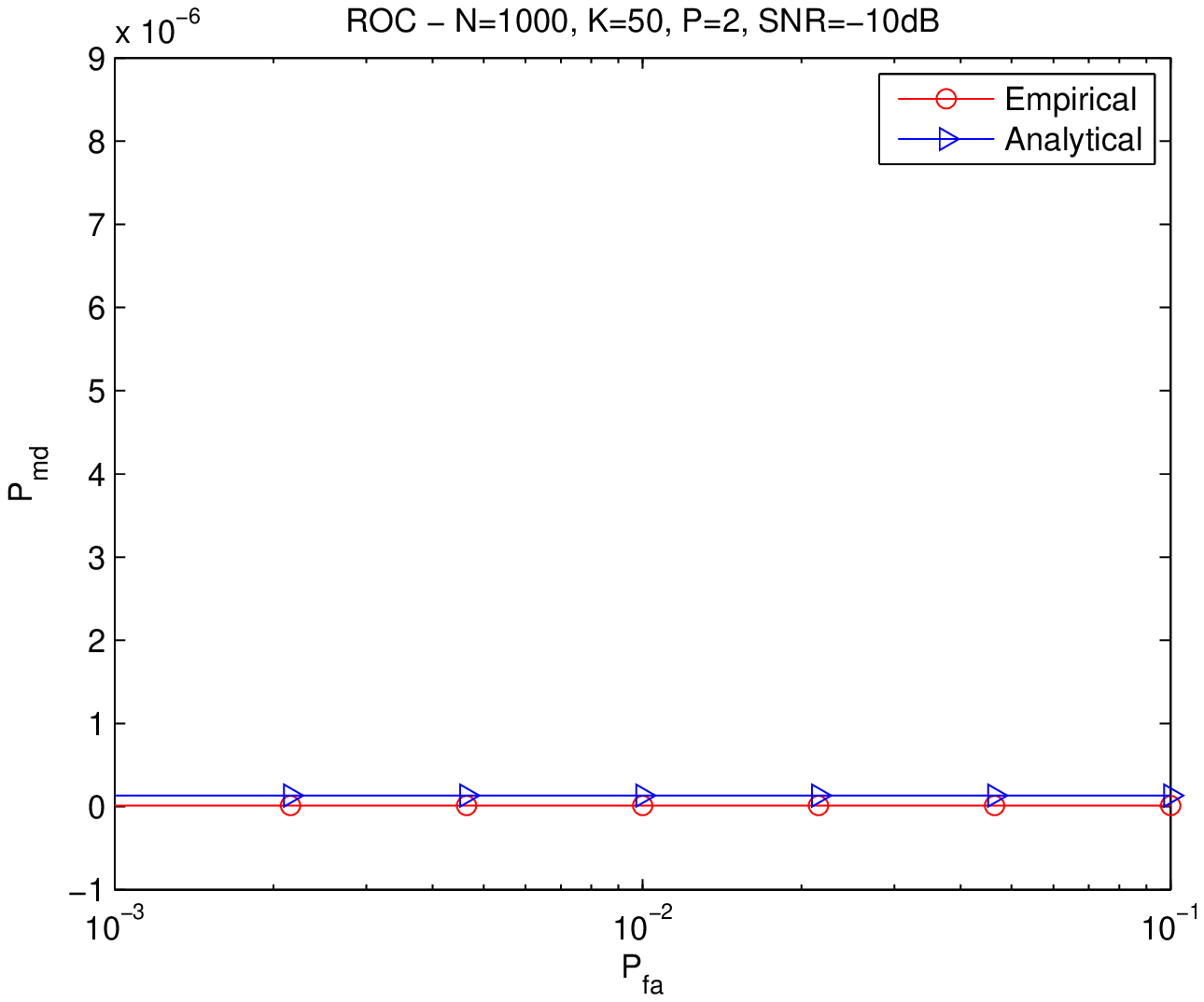}
	\caption{Complementary ROC: analytical vs. empirical. $P=2$, $\rho = -10$ dB.}
	\label{rocP2-10}
\end{figure}

\subsection{Receiver operating characteristics (ROC)}
Figures \ref{rocP1-20} and \ref{rocP2-10} represent the performance of the eigenvalue-based detector in the form of complementary ROC (receiver operating characteristics), i.e., $P_{md}$ as a function of the target $P_{fa}$. The curves are plotted by setting the threshold as a function of the false-alarm probability and deriving the corresponding missed-detection probability for that threshold. The graphs compare the curves obtained from the empirical distributions with those obtained using the analytical expressions of this paper: (\ref{gammapfa}) to set $\gamma(P_{fa})$, then (\ref{Pmd1}) to compute $P_{md}(\gamma)$.

The first ROC graph refers to the same scenario as figures \ref{fth0} (for $P_{fa}$) and \ref{fth1p1snr-20} (for $P_{md}$), with $N=1000$ and $K=50$; the second one refers to the scenario of figure \ref{fth1p2snr-10} (for $P_{md}$) with the same values of $N$ and $K$.

The overall detector performance expressed by the ROC improves as the separation between the $P_{fa}$ curve (monotonically decreasing) and the $P_{md}$ curve (monotonically increasing) gets larger, thus letting both $P_{fa}$ and $P_{md}$ be nearly zero for a wide range of $\gamma$. Such distance increases with $K$, $N$ and with the SNR. For this reason, in the second ROC the performance is almost ideal (zero $P_{md}$ for all the $P_{fa}$). In the first ROC on the contrary there are finite missed-detection probabilities for the considered range of $P_{fa}$; the analytical result also in this case turns out to be consistent with the empirical data.

\section{Conclusion}
\label{concl}
In this paper, analytical formulae have been found for the limiting distribution of 
the ratio between the largest and the smallest eigenvalue in sample covariance matrices, either constructed from pure-noise (Wishart) models or signal-and-noise (spiked population) models. These results have been applied to the problem of signal detection (in particular, in the context of Cognitive Radio), where eigenvalue-based detection has proved to be an efficient technique.


Among the main results of the paper, there are the analytical formulation of the missed detection probability as a function of the threshold, and
     the derivation and discussion of signal identifiability conditions.
     All the results have been validated via numerical simulations covering false-alarm  and missed-detection vs. threshold, convergence behavior, identifiability for single and multiple primary users as a function of the SNR,  validity of the approach for realistic modulated signals, ROC curves.

\appendix

\subsection{Tracy-Widom distribution}
\label{tw_app}
The Tracy-Widom distributions $\mc W_2$ were introduced in \cite{tw}, to express the distribution of the largest eigenvalue in a Gaussian Unitary Ensemble (GUE).
Define the complex Airy function,
\beq
Ai (u) = \frac{1}{2\pi} \int_{\infty e^{5 j \pi/6}}^{\infty e^{j \pi/6}} e^{jua + j \frac{1}{3}a^3} da
\eeq
the Airy kernel,
\beq
A(u,v) = \frac{Ai(u)Ai'(v)-Ai'(u)Ai(v)}{u-v}
\eeq
and let the $\bm A_x$ be the operator acting on $L^2((x, +\infty))$ with kernel $A(u,v)$.
Then, the second-order Tracy-Widom CDF, $F_{\mc W_2}(x)$, is defined in terms of the Fredholm determinant
\beq
F_{\mc W_2}(x) = \mathrm{det}(1- \bm{A}_x)
\eeq
It also admits an alternative expression. Let $q(u)$ be the solution of the Painlev\'e II differential equation
\beq
q''(u) = u q(u) + 2 q^3(u)
\eeq
satisfying
\beq
q(u) \sim -Ai(u), \; \; \; u \ra +\infty
\eeq
Then
\beq
F_{\mc W_2}(x) = \exp \left( -\int_x^{+\infty} (u-x) q^2(u) du \right)
\eeq
Notice that this definition, and the index $2$, are referred to the case of complex Gaussian variables. In the case of real signals, one should use the corresponding first-order Tracy-Widom distribution \cite{tw}.

\subsection{Airy-type distributions}
\label{at_app}
These distributions are defined in \cite{baik_benarous_peche} as an extension of the Tracy-Widom (GUE) distribution.
Let
\begin{align}
s^{(m)}(u) =  \frac{1}{2\pi} \int_{\infty e^{5 j \pi/6}}^{\infty e^{j \pi/6}} e^{jua + j \frac{1}{3}a^3} \frac{1}{(ja)^{m}} da \\
t^{(m)}(u) =  \frac{1}{2\pi} \int_{\infty e^{5 j \pi/6}}^{\infty e^{j \pi/6}} e^{jua + j \frac{1}{3}a^3} (ja)^{m-1} da
\end{align}
Then, for $k \geq 1$, the CDFs of $\mc{A}_k$ are defined as
\begin{align}
F_{\mc{A}_k}(x) = & \; \mathrm{det}(1- \bm{A}_x) \cdot  \\
& \cdot \mathrm{det}\left(\delta_{mn} - < \frac{1}{1- \bm{A}_x} s^{(m)}, t^{(n)} >\right)_{1 \leq m,n \leq k}  \nonumber
\end{align}
where $<,>$ is the real inner product of functions in $L^2((x, +\infty))$.

For $k=0$, this distribution reduces to the GUE distribution: \beq F_{\mc{A}_0}(x) = F_{\mc W_2}(x) \eeq
For $k=1$, it can be written in the Painlev\'e form \beq F_{\mc{A}_1}(x) = F_{\mc W_2}(x) \exp \left( \int_{x}^{+\infty} q(u) du \right) \eeq

\subsection{Finite GUE distributions}
\label{gue_app}
The distributions $\mc G_k$ are defined in \cite{baik_benarous_peche} as the distribution of the largest eigenvalue in a $k \times k$ GUE. Their CDF is
\begin{align}
& F_{\mc G_k} (x) = ( 2 \pi) ^{-k/2} \left( \prod_{m=1}^k m! \right)^{-1} \cdot \\ & \cdot \int_{-\infty}^x \hdots  \int_{-\infty}^x \prod_{1 \leq m <n \leq k} |\xi_m - \xi_n|^2 \cdot \prod_{m=1}^k e^{-\frac{1}{2} \xi_m^2 } d\xi_1 \hdots d\xi_k \nonumber
\end{align}
In the case $k=1$, it is simply a zero-mean, unit-variance Gaussian distribution:
\beq
F_{\mc G_1}(x) = \frac{1}{2\pi} \int_{-\infty}^x e^{-\frac{1}{2} \xi^2} d\xi \triangleq \mc E(x)
\eeq
We also introduce here a compact expression for CDF and PDF in the case $k=2$, in terms of the Gaussian error function:
\begin{align}
F_{\mc G_2}(x) = \mc E^2(x) -  \frac{1}{\sqrt{2\pi}}x e^{-\frac{x^2}{2}} \mc E(x) -  \frac{1}{2\pi} e^{-x^2} \\
f_{\mc G_2}(x) =   \frac{1}{\sqrt{2\pi}} e^{-\frac{x^2}{2}} (1+x^2) \mc E(x) + \frac{1}{2\pi}x e^{-x^2}
\end{align}
These expressions do not appear in \cite{baik_benarous_peche}.


%
%

\ifCLASSOPTIONcaptionsoff
  \newpage
\fi

\end{document}